%% file: main.tex
\documentclass[11pt,notitlepage,tightenlines,nofootinbib,superscriptaddress]{revtex4-2}
\input{packages.tex}

\input{macros.tex}

\begin{document}

\title{Sample complexity of quantum resource testing via one-shot quantum blurring}

\author{Dmitry Grinko}
\affiliation{QuSoft, Amsterdam, The Netherlands}
\affiliation{Institute for Logic, Language and Computation, University of Amsterdam, Amsterdam, the Netherlands}
\affiliation{Korteweg-de Vries Institute for Mathematics, University of Amsterdam, Amsterdam, the Netherlands}

\author{Ludovico Lami}
\affiliation{Scuola Normale Superiore, Piazza dei Cavalieri 7, 56126 Pisa, Italy}

\date{October 30, 2026}

\begin{abstract}
Quantum resource testing is a fundamental primitive of quantum information processing, profoundly connected to resource manipulation. Its goal is to discriminate $n$ copies of a given resourceful state $\rho$ from all free (i.e., resourceless) states; key instances for applications are entanglement testing and quantum magic testing. The asymptotic characterisation relies on the recently proven Generalised Quantum Stein's Lemma (GQSL), which establishes the rate of decay of the false negative error probability for a fixed false positive error probability. This result, however, is intrinsically asymptotic and thus can provide no finite-resource guarantees, which makes its practical implications unclear. Here, we establish the first rigorous finite-$n$ bounds on quantum resource testing and hence quantum resource manipulation, thus strengthening the GQSL and providing explicit estimates on the number of copies needed to achieve a prescribed performance. As notable consequences, we obtain (a)~the convergence of the regularised R\'enyi relative entropies of a resource, which settles the important open problem from~\href{https://ieeexplore.ieee.org/document/11482220}{[Fang/Hayashi, IEEE ToIT 72:6, 2026]}; and (b)~the first sample-complexity bound for asymmetric resource testing: for any fixed false positive error probability, a false negative error probability of at most $\delta$ can be achieved with $n=O\left(\frac{\log(1/\delta)}{D^\infty(\rho\|\FF)}\right)$ copies of $\rho$, in the limit where $\delta \to 0$.
\end{abstract}

\maketitle

\let\oldaddcontentsline\addcontentsline
\renewcommand{\addcontentsline}[3]{}
\fakepart{Main text}

\section{Introduction}

Entanglement is at the same time one of the most puzzling quantum phenomena, which has bewildered scientists for more than a century, and one of the essential ingredients of quantum technologies. For these reasons, its %
detection and manipulation are two of the most fundamental primitives of quantum information science. And yet, in spite of their importance, the ultimate limitations to these tasks remain largely mysterious. 

Furthermore, the rapid progress of quantum technologies has uncovered a whole landscape of different quantum resources, whose detection and manipulation are also critical to different applications. A notable example, besides entanglement, is nonstabiliser-ness, a.k.a.\ `quantum magic', which works as fuel for universal quantum computation~\cite{Bravyi-Kitaev}. To tackle the proliferation of quantum resources, a unified paradigm for their understanding, i.e.\ the formalism of quantum resource theories~\cite{RT-review}, has been developed. In this paper, we will formulate our results in terms of general quantum resources, but all of them apply to entanglement, magic, and many other examples.

In a quantum resource theory, one identifies a set of `free states', whose preparation is assumed to be experimentally inexpensive; any other state is called `resourceful', as it contains some amount of the quantum resource under examination. For example, in entanglement theory separable (a.k.a.\ unentangled) states are considered as free, and any other (entangled) state is resourceful. 
Together with free states, one also specifies a notion of `free operations', namely operations regarded as inexpensive to implement in the laboratory. For consistency, free operations should be `asymptotically resource non-generating' (ARNG), i.e.\ should not inject any additional resource into the system, at least asymptotically.

Quantum resource manipulation is the task of transforming many copies of a given resourceful --- but typically noisy --- quantum state into as many copies as possible of a target --- typically purer --- state by using only free operations. A primary example is entanglement distillation, in which the target is a two-qubit maximally entangled state (a.k.a.\ an `ebit'), which works as fuel for fundamental protocols such as quantum teleportation or quantum key distribution. Resource %
detection, instead, consists in telling whether an unknown state is free or resourceful. 

A key step in the understanding of these two general tasks has been made by Brand\~ao and Plenio~\cite{BrandaoPlenio1, BrandaoPlenio2}, who have recognised that resource manipulation under ARNG operations on the one hand, and resource detection on the other, are in fact two sides of the same coin, and can be rigorously analysed within the mathematical framework of quantum hypothesis testing~\cite{BrandaoGour2015}.
We refer to the resulting mathematical problem, which is the central object of interest here, as \emph{quantum resource testing}~\cite{gap-comment}.
Because of the above equivalence, solving this problem yields a general understanding of both ARNG resource manipulation and resource detection.

While this formalisation is promising, progress has since been hindered by the intrinsic difficulty of the mathematics of quantum resource testing, in turn deeply rooted in the strength of correlations that quantum systems can exhibit. The keystone of the theory is the generalised quantum Stein's lemma (GQSL), which establishes the asymptotic limitations to which both quantum hypothesis testing and quantum resource manipulation are subjected. These are determined by the regularised relative entropy of a resource, which thus acquires a special role as a fundamental quantum resource quantifier~\cite{Vedral1997}. After the initial proof of the GQSL by Brand\~{a}o and Plenio was found to contain a serious gap~\cite{gap, gap-comment}, the result was conclusively proved only recently~\cite{Hayashi-Stein, LamiGQSL}. These solutions, however, have an important drawback: they work only in the asymptotic regime, i.e.\ when the underlying quantum system is composed of $n$ copies of some elementary system, and $n\to \infty$. This is undesirable for applications, for which, instead, one would want to be able to answer the following rather concrete questions: given a fixed error probability threshold, 
\begin{enumerate}[(a), leftmargin=*, itemsep=-8pt, topsep=2pt]
\item how many copies of a quantum system are needed in order to tell whether it contains some quantum resource or not?
\item how many copies of a resourceful quantum state are needed to distill a certain number of copies of some target state with ARNG operations? 
\end{enumerate}

These two questions are essentially equivalent, the link between the two being an instance of the aforementioned general equivalence established by Brand\~{a}o and Plenio~\cite{BrandaoPlenio1, BrandaoPlenio2}. Notably, the intrinsically asymptotic results in~\cite{Hayashi-Stein, LamiGQSL} 
cannot be used to answer them. 
To do that, one needs instead a \emph{fully one-shot} solution of the GQSL, which would yield quantitative estimates on the efficiency of resource detection and manipulation from any concrete (finite) quantum system, as opposed to mere guarantees on the \emph{asymptotic} performance. Indeed, this shift of focus from theoretical, asymptotic results to practically applicable one-shot results has been a major effort throughout the field of quantum information science in the last two decades~\cite{RennerPhD, TOMAMICHEL, KHATRI}.

In the case of the GQSL, this effort has so far been stifled by some serious mathematical obstacles: the information spectrum technique in~\cite{Hayashi-Stein}, as well as the bosonic lifting procedure that is needed to make the quantum blurring work in~\cite{LamiGQSL}, is inherently asymptotic rather than one-shot. As a consequence, the picture of hypothesis testing applied to quantum resources is fundamentally incomplete. In particular, the following key question has been left unanswered~\cite{FangHayashi2025}:
\begin{enumerate}[(a), leftmargin=*, itemsep=-2pt, topsep=2pt] \setcounter{enumi}{2}
\item in quantum resource testing, is it possible to achieve both: a false positive error probability decaying exponentially to zero with $n$, the number of copies of the state; and a false negative error probability decaying exponentially to zero with an optimal exponent arbitrarily close to the regularised relative entropy of the resource?\footnote{Here, a `false positive' error consists in mistaking a resourceful state for a free one, and vice versa for a `false negative' error.}
\end{enumerate}

In this work, we resolve all of the above questions~(a)--(c). We do so by constructing quantitative one-shot estimates on the effectiveness of quantum resource testing --- and thus, in turn, also of quantum resource manipulation. We then leverage these estimates to determine the \emph{sample complexity} of both tasks, thus providing a complete answer to questions~(a)--(b). In the appropriate regime of exponentially decaying false positive and false negative error probabilities, the same estimates also entail an affirmative answer to question~(c). What makes our main technical result particularly appealing is that it is, in information theory parlance, an achievability statement: that is, it amounts to the construction of resource testing schemes with extremely strong error guarantees. This has countless potential technological applications, from the certification of entanglement sources to that of magic state factories. 

From the theoretical point of view, the solution of~(c) also implies that the regularised ($n\to\infty$) $\alpha$-R\'enyi relative entropies of a resource converge, in the limit $\alpha\to 1^-$, to the corresponding regularised relative entropy of a resource. The validity of this statement, which establishes the highly non-trivial commutation of the two limits $\alpha\to 1^-$ and $n\to\infty$, was a major open problem in the field~\cite{FangHayashi2025}. All this is made possible by substantial technical developments, and in particular by a one-shot refinement of the quantum blurring technique of~\cite{LamiGQSL}.

\section{Quantum resource theories}

\subsection{General formalism}

The fundamental ingredient of a quantum resource theory is its associated set of \deff{free states}. 
For a given number of copies $n$ of some elementary quantum system, we thus have an associated set $\FF_n$ of free states.
We will assume that the sets $\FF_n$ obey a series of basic compatibility assumptions, known as Brand\~{a}o--Plenio axioms~\cite{BrandaoPlenio}:
\begin{enumerate}[leftmargin=*, labelindent=4pt, topsep=0pt]
    \item For every $n$, the set $\FF_n$ is closed and convex. 
    
    \item There exists a full-rank state $\tau \in \FF_1$, so that $\tau\ge  c_\tau \frac{\1_A}{d}$ for some constant $c_\tau>0$, where $\1_A/d$ is the maximally mixed state on the $d$-dimensional system $A$.
    
    \item Free states are stable under the operation of discarding subsystems: whenever $\sigma\in \FF_{n+1}$ is free, its marginal on the first $n$ copies is also free, i.e.\ $\Tr_{n+1}\sigma\in \FF_n$, where $\Tr_{n+1}$ denotes the trace over the last copy.
    
    \item Free states can be combined by tensoring: if $\sigma\in \FF_n$ and $\sigma'\in \FF_m$, then $\sigma\otimes\sigma'$ is also free, i.e.\ it belongs to $\FF_{n+m}$.
    
    \item The sets of free states are invariant under relabelling of the copies. More precisely, for every $\sigma\in \FF_n$ and every permutation $\pi\in S_n$, one has
    $U_\pi^{\vphantom{\dag}}\sigma U_\pi^\dag\in \FF_n$, where $U_\pi$ is the unitary that permutes the copies of the system $A$.
\end{enumerate}

Entanglement theory~\cite{Werner} and the resource theory of `quantum magic'~\cite{Veitch2014} are two major examples of quantum resource theories. The former, for instance, corresponds to the choice of $\FF_n$ as the set of separable states, i.e.
\bb
\FF_n = \sep_{A^n:B^n} \defeq \Set*{\sum_x p(x)\ \rho_{A^n}^{(x)}\otimes \sigma_{B^n}^{(x)}} .
\label{eq:separable_states}
\ee
Here, $A$ and $B$ are some elementary quantum systems in Alice's and Bob's laboratories, $x$ is an outcome of a random variable happening with probability $p(x)$, and $\rho_{A^n}^{(x)},\, \sigma_{B^n}^{(x)}$ are collections of arbitrary states on $n$ copies of $A,B$, respectively. In the theory of magic, instead, one takes $\FF_n = \mathrm{STAB}_n$ as the convex hull of so-called stabiliser states.

Essentially all interesting sets of free states, including separable states in entanglement theory and stabiliser states in the resource theory of quantum magic, obey all the Brand\~ao--Plenio axioms~1--5, whose purpose is to make sure that certain elementary operations map free states into other free states. For instance, the convexity requirement (Axiom~1) ensures that the statistical mixture of free states, which corresponds to the procedure of preparing a different free state depending on the outcome of a coin toss, results in another free state. This is guaranteed by construction in~\eqref{eq:separable_states}.

The resource content of a quantum state $\rho$ can be quantified by its \deff{relative entropy of the resource}, mathematically defined by
\bb
D(\rho\|\FF) \defeq \inf_{\sigma\in \FF_1} D(\rho\|\sigma)\, .
\label{eq:relative_entropy_resource}
\ee
Here, $D(\rho\|\sigma) \defeq \Tr [\rho (\log \rho - \log \sigma)]$ is the \deff{quantum relative entropy}, a fundamental statistical distance measure between quantum states~\cite{Umegaki1962, hiai1991proper}. To account for the fact that correlations in free states might lead to better approximations of multiple copies of $\rho$, one needs to \emph{regularise} the simple expression~\eqref{eq:relative_entropy_resource} into
\bb 
D^\infty(\rho\|\FF) \defeq \lim_{n\to\infty} \frac1n \rel{D}{\rho^{\otimes n}}{\FF_n}\, ,
\label{eq:regularised_relative_entropy_resource}
\ee
where we adopted the shorthand $\rel{D}{\rho^{\otimes n}}{\FF_n} \defeq \min_{\sigma_n\in \FF_n} \rel{D}{\rho^{\otimes n}}{\sigma_n}$; this is called the \deff{regularised relative entropy of the resource}.

\subsection{Quantum resource testing}

Given $n$ copies of a quantum system, the task of quantum resource testing consists in designing a measurement to decide between these two hypotheses, one of which is promised to hold: 
\begin{itemize}[leftmargin=*, itemsep=-2pt, topsep=2pt]
\item Null hypothesis: the overall system is in the state $\rho^{\otimes n}$, for some known (resourceful) $\rho$.
\item Alternative hypothesis: the overall system is in a free state, i.e.\ it contains no resource.
\end{itemize}
Two different kinds of errors are possible: 
\begin{itemize}[leftmargin=*, itemsep=-2pt, topsep=2pt]
\item Type I errors (false positives): we reject the null hypothesis even though it was true.
\item Type II errors (false negatives): we accept the null hypothesis even though it was false.
\end{itemize}

Depending on the specifics of the problem, the measurement on the quantum system can be designed to minimise one of the two corresponding error probabilities, or to achieve a suitable tradeoff between them. A key role in the theory is played by the \deff{minimal type~II error probability} for a given threshold $\e\in (0,1)$ on the type~I error probability, denoted by
\bb
\rel{\beta_\varepsilon}{\rho^{\otimes n}}{\FF_n} \defeq \min\Set{ \pr\set{\text{type~II error}}\, \given \, \pr\set{\text{type~I error}} \leq \e}\, ;
\ee
here, the minimisation is over all possible measurements. For a given $\delta\in (0,1)$, the \deff{sample complexity} of quantum resource testing is the minimum number of samples needed to achieve a type~II error probability at most equal to $\delta$. Formally,
\bb
N_{\e,\delta}(\rho\|\FF) \defeq \min\Set*{n\, \given\, \rel{\beta_\varepsilon}{\rho^{\otimes n}}{\FF_n}\leq \delta}
\label{eq:sample-complexity}
\ee
A precise estimate of the above quantity would yield a quantitative answer to question~(a) in the introduction.

\subsection{Connection with resource manipulation}

The main conceptual contribution of Brand\~ao and Plenio~\cite{BrandaoPlenio1, BrandaoPlenio2} is the realisation that the task of quantum resource testing is fully equivalent to quantum resource manipulation under ARNG operations~\cite{BrandaoGour2015}.
This connection is particularly simple and illuminating in the case of entanglement theory.
In that setting, ARNG operations can be replaced with the even simpler class of non-entangling (NE) operations, and the Brand\~ao--Plenio connection establishes that \emph{it is possible to distil a $d$-dimensional maximally entangled state from $\rho_{AB}$ using NE operations with trace distance error $\e$ if and only if $\beta_\e(\rho_{AB}\,\|\, \sep_{A:B})\leq 1/d$}~\cite[proof of Theorem~1]{BrandaoDatta2011}.
Due to this connection, questions~(a) and~(b) in the introduction are completely equivalent.

\subsection{Error exponents}

The traditional approach to quantum resource testing involves the study of the asymptotic behaviour of error probabilities as $n\to\infty$. A key role is played by the \deff{Stein exponent} $\stein(\rho\|\FF)$, defined to be the largest $r\ge  0$ such that it is possible to make the type~II error probability decay to zero exponentially fast with exponent $r$, i.e.\ $\rel{\beta_{\varepsilon_n}}{\rho^{\otimes n}}{\FF_n}\sim 2^{-rn}$, while at the same time the type~I error vanishes asymptotically, i.e.\ $\e_n \smash{\underset{\text{\raisebox{8pt}{$n\!\to\!\infty$}}}{\longrightarrow}} 0 \vphantom{\Big|}$. The GQSL shows that the Stein exponent coincides, for all states $\rho$, with the regularised relative entropy of the resource~\cite{BrandaoPlenio, Hayashi-Stein, LamiGQSL}:
\bb
\stein(\rho\|\FF) = D^\infty(\rho\|\FF)\, .
\ee

A finer understanding can be obtained by looking also at the speed at which $\e_n$ vanishes asymptotically. When $\e_n = 2^{-sn}$ vanishes exponentially fast, for some coefficient $s>0$ that we think of as sufficiently small,\footnote{More precisely, one needs to have $0<s<D(\FF_1\|\rho) \defeq \min_{\sigma\in \FF_1} D(\sigma\|\rho)$, where the latter quantity is known as the reverse relative entropy of the resource.} it has been shown that~\cite{FangHayashi2025}
\bb
\rel{\beta_{2^{-ns}}}{\rho^{\otimes n}}{\FF_n} \sim 2^{-r(s)\, n},\qquad r(s) = \lim_{n\to\infty} \sup_{0<\alpha<1} \left\{ \frac{1}{n} \rel{D_\alpha}{\rho^{\otimes n}}{\FF_n} - \frac{\alpha}{1-\alpha}\,s \right\},
\ee
where $D_\alpha(\rho\|\sigma) \defeq \frac{1}{\alpha-1} \log \Tr[\rho^\alpha \sigma^{1-\alpha}]$ is the so-called Petz--R\'enyi relative entropy~\cite{PetzRenyi}, and, as before, we used the shorthand notation $\rel{D_\alpha}{\rho^{\otimes n}}{\FF_n} \defeq \min_{\sigma_n\in \FF_n} \rel{D_\alpha}{\rho^{\otimes n}}{\sigma_n}$.

By construction, $r(s) \leq \stein(\rho\|\FF)$ needs to hold for all $s>0$, as the right-hand side is defined by making no promises on the speed of convergence of $\e_n$ to zero. Physically, however, we expect this inequality to become an equality in the limit $s\to 0^+$: if this were really the case, we would immediately obtain an affirmative answer to question~(c) in the introduction. A simple calculation shows that this problem is equivalent to a limit commutation problem:
\bb
\lim_{\alpha\to 1^-} \lim_{n\to\infty} \frac{1}{n} \rel{D_\alpha}{\rho^{\otimes n}}{\FF_n} \eqt{?} \lim_{n\to\infty} \lim_{\alpha\to 1^-} \frac{1}{n} \rel{D_\alpha}{\rho^{\otimes n}}{\FF_n} = D^\infty(\rho\|\FF)
\label{eq:limit_commutation_problem}
\ee
Prior to our work, the difficulty of handling the regularisation limit $n\to\infty$ had precluded a solution to the above question. The difficulty is somewhat analogous to the mechanism underlying phase transitions in statistical physics: the limit $n\to\infty$ plays the role of the thermodynamic limit, while $\alpha$ acts as a system parameter, and a possible `phase transition' at $\alpha=1$ may prevent the two limits from being exchanged.

\section{Results}

Our main result establishes the first one-shot achievability bounds on quantum resource testing:

\begin{theorem}[(One-shot formulation of the GQSL)]
\label{thm:main-one-shot-gqsl}
Consider a quantum resource theory defined by the free states $\FF = (\FF_n)_n$, assumed to satisfy all Brand\~ao--Plenio axioms. Then, there are constants $N,A,B>0$ and $\e_0\in (0,1)$ with the following property: for all $n\ge  N$, all type~I error thresholds $2^{-n/240} \leq \e \leq \e_0$, and all states $\rho$, we have
\bb
- \log \rel{\beta_{\e}}{\rho^{\otimes n}}{\FF_n} \ge  \rel{D}{\rho^{\otimes n}}{\FF_n} - A n \sqrt{\e} - B \sqrt{n \log\frac1\e}\, .
\label{eq:one_shot_GQSL}
\ee
\end{theorem}

The above result is an achievability statement: namely, it establishes the existence of measurements that detect quantum resources very effectively, i.e.\ with stringent limits on the associated error probabilities. More precisely, for a given type I (false positive) error probability threshold $\e$, within the above range, \cref{eq:one_shot_GQSL} guarantees that we can find a measurement detecting quantum resources with type II (false negative) error probability $\beta(n) \leq 2^{-D(\rho^{\otimes n}\|\FF_n)}\,\!\cdot\,\! 2^{An\sqrt\e + B\sqrt{n\log(1/\e)}}$. In writing out this expression, we distinguished two contributions. (i)~The first factor determines the dominant scaling of the type II error probability, governed by the Stein exponent: indeed, for large $n$ we have $D(\rho^{\otimes n}\|\FF_n) \approx n D^\infty(\rho\|\FF) = n\,\stein(\rho\|\FF)$, thus recovering the asymptotic version of the GQSL. (ii)~The second factor is a penalty term, which quantifies precisely the difference between the $n$-copy performance and the asymptotic benchmark. 

By considering different regimes in our key estimate, \cref{eq:one_shot_GQSL}, we can now successfully answer questions~(a)--(c) in the introduction. We start from question~(c), which has an affirmative answer:

\begin{theorem}[(Solution to the limit commutation problem)]
\label{thm:main-limit-commutation}
In a quantum resource theory that obeys all Brand\~ao--Plenio axioms, question~(c) in the Introduction has an affirmative answer. Equivalently:
\begin{enumerate}[(i)]
\item for all states $\rho$ and all $r < D^\infty(\rho\|\FF)$, where the right-hand side is defined by~\eqref{eq:regularised_relative_entropy_resource}, we can find a sufficiently small $s>0$ such that $\rel{\beta_{2^{-sn}}}{\rho^{\otimes n}}{\FF_n} \leq 2^{-rn}$ for all sufficiently large $n$;
\item the identity in~\eqref{eq:limit_commutation_problem} holds: $\lim_{\alpha\to 1^-} \lim_{n\to\infty} \frac{1}{n} \rel{D_\alpha}{\rho^{\otimes n}}{\FF_n} = D^\infty(\rho\|\FF)$.
\end{enumerate}
\end{theorem}

The above result resolves the main open problem from~\cite{FangHayashi2025}, showing that it is possible to recover the asymptotic performance of quantum resource testing established by the GQSL of~\cite{Hayashi-Stein, LamiGQSL} from the error exponent picture of~\cite{FangHayashi2025}. This effectively makes the two frameworks click together perfectly, finally yielding a unified, coherent picture of quantum resource testing.

Theorem~\ref{thm:main-limit-commutation} can be proved by looking at \cref{eq:one_shot_GQSL} in the regime where $\e = 2^{-sn}$ is exponentially suppressed. By considering instead a fixed $\e$, we can now answer questions~(a)--(b):

\begin{theorem}[(Sample complexity of quantum resource testing)]
\label{thm:main-sample-complexity}
In a quantum resource theory that obeys all Brand\~ao--Plenio axioms, and for all fixed 
$\e>0$, the sample complexity of resource testing, defined by~\eqref{eq:sample-complexity}, satisfies
\bb
N_{\e,\delta}(\rho\|\FF) = O\left(\frac{\log(1/\delta)}{D^\infty(\rho\|\FF)}\right) \qquad (\delta\to 0^+)\, .
\ee

Hence, to distill $k$ ebits with fixed error using non-entangling operations one needs $O\left(k/E(\rho)\right)$ 
copies of a bipartite state $\rho = \rho_{AB}$, where $E(\rho) = D^\infty(\rho\,\|\,\sep)$ is the regularised relative entropy of entanglement.
\end{theorem}

The above theorem confirms our physical intuition: if the type II error in resource testing scales as $\sim 2^{-n D^\infty(\rho\|\FF)}$, then $n\sim \frac{\log(1/\delta)}{D^\infty(\rho\|\FF)}$ copies of the state should suffice to make it fall below a given threshold $\delta>0$. The problem with this intuitive reasoning, and the reason a rigorous analysis is needed to establish the above result, is that prior work only guarantees that the above exponential behaviour appears \emph{asymptotically}\,---\,it does not tell us at which point, i.e.\ for which $n$. Theorem~\ref{thm:main-sample-complexity}, instead, provides explicit estimates on $n$ and is thus much more suited for practical applications.

\section{Discussion}

We have established rigorous one-shot estimates for quantum resource testing and manipulation, pinpointing how many copies of a given quantum state are needed to achieve a prescribed error threshold. The asymptotic performance of these two tasks is captured by the generalised quantum Stein's lemma, and our main bound can therefore be regarded as a one-shot formulation of that result. Our work thus marks another step in the broader ongoing programme of characterising quantum information tasks directly in the finite blocklength regime, rather than solely through their asymptotic performance.

We have resolved the main open problem from~\cite{FangHayashi2025}, proving that the two error probabilities in quantum resource testing can be made to decrease exponentially at the same time. If the type~I error scales as $\sim 2^{-sn}$, with some small exponent $s>0$, the type~II error (Stein) exponent can be chosen to be as close to $D^\infty(\rho\|\FF)$ as desired, provided that $s$ is sufficiently small. Thus, exponential suppression of the type~I error entails no residual loss in the Stein exponent as $s \to 0^+$. This immediately settles the limit commutation problem in~\eqref{eq:limit_commutation_problem}, showing that the regularised Petz--R\'enyi relative entropy of resource is continuous from below at the Umegaki `critical point' $\alpha=1$. The error exponent formula of~\cite{FangHayashi2025} and the generalised quantum Stein's lemma of~\cite{Hayashi-Stein, LamiGQSL} therefore yield the same rate when the type~I exponent tends to zero. 
This shows that the examples of discontinuity 
constructed in~\cite{FangHayashi2025} are, in some sense, pathological: most physically motivated quantum resource theories, including entanglement and magic, obey the Brand\~ao--Plenio axioms; in those, due to Theorem~\ref{thm:main-limit-commutation}, no discontinuity can appear. 

Our results also show that, at a fixed type~I tolerance, the sample complexity for reducing the type~II probability below $\delta$ is of order $O\left(\log(1/\delta)/D^\infty(\rho\|\FF)\right)$ as $\delta \to 0^+$. This endows the regularised relative entropy of the resource, $D^\infty(\rho\|\FF)$, with a further operational interpretation, as the scale that determines the necessary experimental resources in asymmetric quantum resource testing. Through the Brand\~ao--Plenio correspondence, this also entails a finite-copy guarantee for resource conversion under asymptotically resource non-generating operations. In entanglement theory, for example, $k$ ebits can be produced with constant error under non-entangling operations from $O\left(k\big/D^\infty(\rho\|\sep)\right)$ copies of any input state $\rho$. This is the kind of information that asymptotic rates alone cannot provide, and it supplies a quantitative benchmark for assessing, e.g., entangled state production devices and magic state factories.

Our proof is based on a one-shot version of the quantum blurring technique put forth in~\cite{LamiGQSL}.
`Blurring' refers to the controlled addition of a small amount of noise to the underlying quantum system, whose goal is to make it more regular and better behaved.
The original technique, however, relied on an asymptotic limit to control the effect of blurring, and could therefore provide no finite-$n$ guarantees.
We circumvent this fundamental conceptual obstacle by combining Tikhonov regularisation with a discrete polynomial localisation technique to pinpoint the effects of blurring for all finite $n$.

Taken together, our results constitute a finite-resource formulation of the theory of quantum resource detection and manipulation, and therefore bring us one step closer to experimental and technological applications of these fundamental quantum information processing primitives.

\medskip
\textbf{Note added.} After we had derived the main result of this manuscript, we realised that qualitatively similar estimates to ours could be derived by further developing some techniques by Mazzola, Sutter, and Renner~\cite{MazzolaSutterRenner2026}. We present these derivations in the Supplementary Information (Section~\ref{sec:alternative_proof_MSR}). Remarkably, however, these alternative estimates are quantitatively less accurate than ours.

\medskip
\textbf{Acknowledgements.} We thank Marco Tomamichel for first raising the question of exponential convergence below rate for the GQSL, and Felix Leditzky, Filippo Girardi, and Kun Fang for useful discussions. The open problem that we resolve here was also discussed at the conference ``Beyond IID in Information Theory 14'' (Shenzhen, 22--26 June 2026). We acknowledge the assistance of ChatGPT and Codex (OpenAI) in developing proofs of several technical lemmas and in refining the exposition. More specifically, after the proofs of main theorems were reduced to a few critical estimates and we had identified some key lemmas, GPT 5.5 was used to search the literature for techniques and contributed nontrivial insights, which served as a basis for \cref{lem:clean-discrete-delta-approximant,lem:clean-hypergeometric-expansion,lem:clean-diagonal-kraus-approximant}. The central idea underlying this project---namely, to use the blurring channel with the maximally mixed state together with the variational form of Tikhonov regularisation---was conceived by the authors. All statements and proofs were derived and verified by the authors. LL acknowledges financial support from the European Union (ERC StG ETQO, Grant Agreement no.\ 101165230). DG acknowledges support by NWO grant NGF.1623.23.025 (“Qudits in theory and experiment”). DG is also grateful to the Scuola Normale Superiore for its hospitality during his visit (4--15 May 2026), when most of this work was carried out.

\bibliography{refs}

\let\addcontentsline\oldaddcontentsline

\clearpage
\fakepart{Supplementary Information}

\onecolumngrid
\begin{center}
\vspace*{\baselineskip}
{\textbf{ {\Large Supplementary Information:} \\[1.5ex] {\large Sample complexity of quantum resource testing via one-shot quantum blurring} }}
\end{center}

\renewcommand{\thesection}{S.\Roman{section}}
\renewcommand{\theequation}{S\arabic{equation}}
\renewcommand{\thetheorem}{S\arabic{theorem}}
\renewcommand{\thelemma}{S\arabic{theorem}}
\renewcommand{\thecorollary}{S\arabic{theorem}}
\renewcommand{\theproposition}{S\arabic{theorem}}
\renewcommand{\thedefinition}{S\arabic{definition}}
\renewcommand{\theremark}{S\arabic{remark}}
\renewcommand{\thefigure}{S\arabic{figure}}
\setcounter{section}{0}
\makeatletter

\setcounter{secnumdepth}{2}

\tableofcontents

\section{The setup} \label{sec:setup}

\subsection{Quantum divergences}

All logarithms are base two unless explicitly denoted by $\ln$. Let $\HH$ be a finite-dimensional Hilbert space, with $d\defeq\dim\HH$. We will denote by $\D\of{\HH}$ the set of density operators on $\HH$, i.e.\ the set of positive semidefinite operators with trace one. For two states $\rho,\sigma$, their Umegaki \deff{relative entropy} is defined by~\cite{Umegaki1962, hiai1991proper}
\bb \label{eq:relative-entropy}
  D\of{\rho\|\sigma}
   \defeq 
  \begin{cases}
    \Tr\rho\of{\log\rho-\log\sigma},
      & \supp\rho\subseteq\supp\sigma,\\
    +\infty,
      & \text{otherwise}.
  \end{cases}
\ee
The \deff{max-relative entropy}~\cite{Datta2009} and its purified-distance smoothing, called the \deff{smooth max-relative entropy}, are given by
\begin{align}
  D_{\max}\of{\rho\|\sigma}
  & \defeq 
  \inf\Set*{\lambda\in\R \given \rho\le2^\lambda\sigma},
  \label{eq:Dmax-definition}\\
  D_{\max}^{\varepsilon}\of{\rho\|\sigma}
  & \defeq 
  \inf_{\rho':\,P\of{\rho,\rho'}\le\varepsilon}
  D_{\max}\of{\rho'\|\sigma},
  \label{eq:smooth-Dmax-definition}
\end{align}
where $P\of{\omega,\tau}=\sqrt{1-F\of{\omega,\tau}^2}$ is the \deff{purified distance}~\cite{TOMAMICHEL}, with $F\of{\omega,\tau} \defeq \left\| \sqrt{\omega} \sqrt{\tau} \right\|_1$ being the fidelity, and the infimum in~\eqref{eq:smooth-Dmax-definition} ranges over normalised states.

For a set of states $\FF \subseteq \D\of{\HH}$, some $\rho\in \D\of{\HH}$, and a general function $\mathbb{D}:\D\of{\HH}\times \D\of{\HH} \to \R \cup \{+\infty\}$, we use the convention
\bb
\mathbb{D}\of{\rho \| \FF} \defeq \inf_{\sigma\in \FF} \mathbb{D}\of{\rho \| \sigma}\, .
\ee

\subsection{One-shot quantum hypothesis testing}

We will now describe the task of quantum hypothesis testing in general terms. We start by looking at the most elementary case, that of discrimination between two simple hypotheses: given two states $\rho,\sigma \in \D\of{\HH}$ and a quantum system modelled by $\HH$ in an unknown state $\omega$, we need to guess between the two following hypotheses, given the promise that one of them holds true:
\begin{itemize}
\item $\omega = \rho$ (null hypothesis); and
\item $\omega = \sigma$ (alternative hypothesis).
\end{itemize}
To guess, we are allowed to measure the single copy of $\omega$ that we are given. We assume that we can implement any physically realisable quantum measurement on $\HH$. 

Two types of error are possible: a type~I error, which consists in guessing the alternative hypothesis when the null hypothesis holds, and a type~II error, which, vice versa, consists in guessing the null hypothesis when the alternative hypothesis holds. A key quantity for the quantitative analysis of this task is the minimum type~II error probability given a threshold $\e \in[0,1]$ on the type~I error probability. Its negative logarithm is the hypothesis testing relative entropy, given by~\cite{Buscemi2010}
\bb
\label{eq:DH-definition}
D_{\mathrm H}^{\e}\of{\rho\|\sigma} \defeq -\log\inf\Set*{ \Tr\sigma T \given 0\le T\le\1,\ \Tr\rho T\ge1-\e} .
\ee

The above setting is usually referred to as \emph{simple} hypothesis testing, due to the fact that the two hypotheses are modelled by one quantum state each. In several applications, however, and in this paper in particular, at least one of the two hypotheses is modelled by a \emph{set} of states. The corresponding discrimination task is therefore known as \emph{composite} hypothesis testing. In this paper, the null hypothesis will be taken to be simple, but the alternative hypothesis will be, in general, composite. The setting is therefore described by a state $\rho\in \D\of{\HH}$ and a set $\FF\subseteq \D\of{\HH}$. One is handed an unknown state $\omega\in\D\of{\HH}$ on the same system, and needs to make a measurement to guess between:
\begin{itemize}
\item $\omega = \rho$ (null hypothesis); and
\item $\omega \in \FF$ (alternative hypothesis).
\end{itemize}
In this paper, $\FF$ will typically represent the set of states that are `easy to prepare' in some operational setting modelled by a quantum resource theory~\cite{RT-review}. For this reason, we shall refer to $\FF$ as the set of \deff{free states}. For example, we can think of $\FF$ as the set of separable (i.e.\ unentangled) states on some bipartite quantum system --- more on that below.

In the above hypothesis testing task, the minimal type~II error probability for a given threshold $\e\in [0,1]$ on the type~I error probability is given by
\bb
\beta_\varepsilon\of{\rho\|\FF} \defeq \inf_{\substack{0\le T\le\1\\ \Tr\rho T \ge1-\varepsilon}} \sup_{\sigma \in \FF}\Tr\sigma T\, ,
\ee
where the inner supremum indicates that we consider the worst-case error over all $\sigma\in \FF$, adversarially chosen given the test $T$. Due to Sion's minimax theorem, if $\FF$ is a convex set then we can write (see, e.g.,~\cite[Lemma~31]{Fang2025})
\bb
\dhh{\e}\of{\rho\|\FF} = -\log \beta_\e\of{\rho\|\FF}\, .
\ee

\subsection{Asymptotic hypothesis testing}

We will be interested in the asymptotic behaviour of quantum hypothesis testing. Namely, we will look at \emph{sequences} of composite hypothesis testing problems. Given some $\rho\in \D\of{\HH}$ and a sequence $\FF = (\FF_n)_n$ of sets of free states $\FF_n\subseteq \D\of{\HH^{\otimes n}}$, the hypothesis testing problem at the $n^\text{th}$ level is guessing whether a state $\omega_n\in \D\of{\HH^{\otimes n}}$ satisfies
\begin{itemize}
\item $\omega_n = \rho^{\otimes n}$ (null hypothesis); or
\item $\omega_n \in \FF_n$ (alternative hypothesis).
\end{itemize}

The asymptotics of the error probabilities is described by the \deff{error exponents}. The simplest one is the \deff{Stein exponent}, which corresponds to the maximum $r\ge  0$ such that the type~II error probability for a fixed (and very small) type~I error probability threshold $\e>0$ can be made to decay as $\sim 2^{-rn}$ as $n\to\infty$. Formally, we set
\bb
\stein(\rho\|\FF) \defeq \lim_{\e\to 0^+} \liminf_{n\to\infty} \frac1n\, \rel{\dhh{\e}}{\rho^{\otimes n}}{\FF_n}\, .
\ee
We will however be interested in a more fine-grained understanding of the error behaviour. Specifically, we will be interested in controlling simultaneously the error exponents governing the decay of both types of error simultaneously. To this end, for some $s>0$ we define the \deff{reverse Hoeffding exponent} as the maximum $r\ge  0$ such that, for all $s>s'>0$, we can simultaneously achieve a decay $\sim 2^{-ns'}$ for the type~I error probability and a decay $\sim 2^{-nr}$ for the type~II error probability. Formally,
\bb
\revh_s(\rho\|\FF) \defeq \lim_{t \to 0^+} \liminf_{n\to\infty} \frac1n\, \rel{\dhh{2^{-n(s-t)}}}{\rho^{\otimes n}}{\FF_n}\, .
\ee

Note that $\revh_s(\rho\|\FF)$ is monotonically non-increasing in $s>0$, and 
\bb
\revh_s(\rho\|\FF) \leq \stein(\rho\|\FF) \qquad\forall\ s>0\, ,
\ee
due to the fact that $2^{-n(s-t)} \leq \e$ for all $\e>0$ and all sufficiently small $t>0$, provided that $n$ is sufficiently large. For many `non-pathological' sequences $\FF$, it also holds that
\bb
\lim_{s\to 0^+} \revh_s(\rho\|\FF) = \stein(\rho\|\FF)\, ;
\label{eq:reverse_hoeffding_convergence}
\ee
however, establishing the above equality without a closed-form expression for $\revh_s(\rho\|\FF)$ is, in general, difficult. One of the main contributions of this paper is to show that~\eqref{eq:reverse_hoeffding_convergence} holds for all sets of free states obeying some natural axioms, and in particular for the set of separable states.

The \deff{Hoeffding exponent} is the inverse function of $s\mapsto \revh_s(\rho\|\FF)$; it is defined by
\bb
\hoeff_r(\rho\|\FF) \defeq \sup \left\{ s>0:\ \revh_s(\rho\|\FF)\ge  r \right\} .
\ee

\subsection{Brand\~ao--Plenio axioms}

Physically meaningful sequences of free-state sets $\FF_n\subseteq\D\of{\HH^{\otimes n}}$ usually obey some basic compatibility assumptions, formalised in~\cite{BrandaoPlenio} and henceforth referred to as \deff{Brand\~ao--Plenio axioms}. They can be written as follows:
\begin{align}
\forall\ n\in \N,\ \FF_n &\text{ is closed and convex,} \label{ax:closed-convex} \\
\exists\ \tau\in\FF_1 :\ \ \tau&\ge c_\tau\,\frac{\1_{\HH}}{d} \ \ \text{for some }c_\tau\in(0,1]\, , \label{ax:free-full-rank} \\
\sigma_n\in\FF_n &\Longrightarrow \Tr_n \sigma_n\in\FF_{n-1}\, ,\label{ax:trace} \\
\sigma_n\in\FF_n,\ \sigma'_m\in\FF_m &\Longrightarrow \sigma_n\otimes\sigma'_m\in\FF_{n+m}\, , \label{ax:tensor} \\
\sigma_n\in \FF_n\, ,\ \pi\in S_n &\Longrightarrow U_\pi^{\vphantom{\dag}} \sigma_n U_\pi^\dag \in \FF_n\, . \label{ax:permutation-invariant}
\end{align}
Here, 
$\Tr_n$ in~\eqref{ax:trace} indicates the partial trace over the last subsystem. By \eqref{ax:tensor} and \eqref{ax:free-full-rank},
\begin{equation}\label{eq:free-full-rank-tensor-power}
  \tau^{\otimes n}\in\FF_n,
  \qquad
  \tau^{\otimes n}
  \ge
  c_\tau^n
  \frac{\1_{\HH^{\otimes n}}}{d^n}.
\end{equation}

For a fixed state $\rho\in\D\of{\HH}$, define
\begin{equation}\label{eq:D-infty}
  D^\infty\of{\rho\|\FF}
   \defeq 
  \lim_{n\to\infty}
  \frac1n\inf_{\sigma_n\in\FF_n}
  D\of{\rho^{\otimes n}\|\sigma_n},
\end{equation}
which exists by subadditivity~\cite{Fekete1923}.
It is finite because $\tau^{\otimes n}\in\FF_n$ and
\begin{equation}
  \frac1n
  \inf_{\sigma_n\in\FF_n}
  D\of{\rho^{\otimes n}\|\sigma_n}
  \le
  \frac1n
  D\of{\rho^{\otimes n}\|\tau^{\otimes n}}
  \le
  \log(d/c_\tau).
\end{equation}

For a test $T$, write
\begin{equation}
  \alpha_n\of{T}
  \defeq
  1-\Tr\rho^{\otimes n}T,
  \qquad
  \beta_n\of{T}
  \defeq
  \sup_{\sigma_n\in\FF_n}\Tr\sigma_nT.
\end{equation}
Thus $D_{\mathrm H}^{\varepsilon}$ is the negative logarithm of the smallest type-II error attainable under the constraint $\alpha_n\of{T}\le\varepsilon$.
The usual Stein regime fixes $\varepsilon>0$ before taking $n\to\infty$; in that fixed-error regime the asymptotic type-II exponent is independent of $\varepsilon$.
The strengthened statement below concerns a different axis: it allows the type-I constraint itself to be exponentially small, $\varepsilon_n=2^{-s n}$, and asserts that no exponent is lost when $s\to0$ after the asymptotic limit.

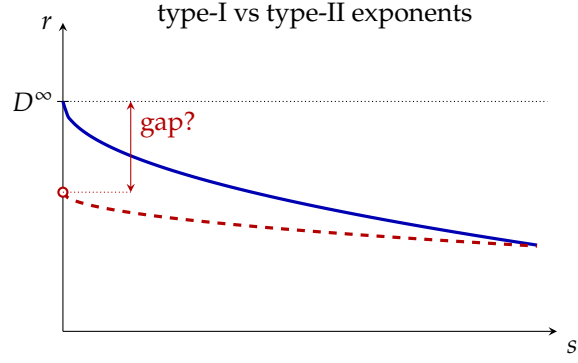
\begin{figure}[t]
  \centering
  \begin{tikzpicture}[>=stealth,x=6.4cm,y=4.0cm]
    \draw[->] (0,0) -- (1.02,0)
      node[below right,font=\small] {$s$};
    \draw[->] (0,0) -- (0,1.02)
      node[left,font=\small] {$r$};
     \node[below,font=\small] at (0.52,1.12)
       {type-I vs type-II exponents};

    \draw[densely dotted] (0,0.76) -- (1.00,0.76);
    \draw (-0.012,0.76) -- (0.012,0.76);
    \node[left,font=\small] at (0,0.76) {$D^\infty$};

    \draw[very thick,blue!70!black]
      plot[domain=0:0.98,samples=100,smooth]
      (\x,{0.76-0.48*sqrt(\x)});

    \draw[very thick,dashed,red!70!black]
      plot[domain=0:0.98,samples=100,smooth]
      (\x,{0.46-0.18*sqrt(\x)});
    \draw[red!70!black,fill=white,thick] (0,0.46) circle (1.7pt);
    \draw[densely dotted,red!70!black] (0,0.46) -- (0.14,0.46);
    \draw[<->,red!70!black] (0.14,0.46) -- (0.14,0.76)
      node[midway,above right,font=\small] {gap?};
  \end{tikzpicture}
  \caption{How \cref{thm:supp-limit-commutation} clarifies the tradeoff between type-I and type-II errors.
  The horizontal axis is the exponent $s$ in the type-I error requirement $\alpha_n\le2^{-s n}$.
  The vertical axis is the corresponding asymptotic type-II exponent.
  At $s=0$ one recovers the usual fixed-error Stein setting, whose exponent is $D^\infty$.
  The blue curve depicts the square-root profile $D^\infty-K_{d,\tau}\sqrt{s}$ established by our lower bound: forcing the type-I error to vanish exponentially may lower the type-II exponent for fixed $s>0$, but the loss disappears as $s\to0$.
  The dashed red curve shows the hypothetical scenario if our result did not hold: the ordinary fixed-error Stein lemma could still hold, while the corresponding type-II exponents approach a strictly smaller value.
  \cref{thm:supp-limit-commutation} rules out this gap quantitatively by giving the lower bound $D^\infty-K_{d,\tau}\sqrt{s}$ for sufficiently small $s>0$, with $K_{d,\tau}=16\of*{2\log d+\log(1/c_\tau)+1}$.
  Here $c_\tau$ is the constant in $\tau\ge c_\tau\1_{\HH}/d$.}
  \label{fig:stein-region}
\end{figure}

\subsection{The blurring channel} \label{subsec:blurring-definition}

Let $\HH$ be a finite-dimensional Hilbert space, with $D=\dim\HH$, and fix $0\le k\le n$.
For a subset $J\subseteq\set{1,\ldots,n}$ with $\abs{J}=k$, let $\Tr_J$ denote the partial trace over the tensor factors with index in $J$. For a state $\omega\in \D\of{\HH}$, define the \deff{blurring channel}
\bb \label{eq:blurring-channel}
\BB_{n,k}^\omega\of{X} \defeq \binom nk^{-1}\sum_{\substack{J\subseteq\set{1,\ldots,n} \\ \abs{J}=k}} \left(I_{J^c} \otimes \big(\omega^{\otimes k} \Tr\big)_J\right) \of{X}\, ,
\ee
where the replacer channel acts on the sites in $J$, and on the sites in $J^c$ we act with the identity. We will typically fix $\omega$ to be either the free state $\tau\in \FF_1$ whose existence is guaranteed by~\eqref{ax:free-full-rank}, or the maximally mixed state $\1_\HH/D$. In the latter case, we omit the superscript:
\bb \label{eq:blurring-with-identity}
\BB_{n,k} \defeq \BB^{\1_\HH/D}_{n,k}\, .
\ee
Since $\tau\ge c_\tau\1_{\HH}/d$, every $X\ge0$ obeys
\bb \label{eq:tau-blurring-dominates-maximally-mixed}
\BB^{\tau}_{n,k}\of{X} \ge c_\tau^k\,\BB_{n,k}\of{X}\, .
\ee
Moreover,
\bb \label{eq:blurring-preserves-free}
\BB^\tau_{n,k}\of{\FF_n} \subseteq \FF_n\, ,
\ee
by the partial-trace, tensor-product, permutation-invariance, and convexity axioms.

\subsection{Symmetric subspace}

Let $\HH$ be a $D$-dimensional Hilbert space with basis $\ket{1},\ldots, \ket{D}$. For some positive integer $n$, the \deff{symmetric subspace} $\Sym^n\of{\HH}$ is defined as the set of all vectors on $\HH^{\otimes n}$ that are invariant under any permutation of the tensor factors. A natural basis of $\Sym^n\of{\HH}$ is the \deff{Dicke basis}; to define it, consider some `weight' $\alpha\in \N^D$, and set~\cite[Section~1]{Harrow-church}
\bb
\ket{\alpha} \coloneqq \binom{n}{\alpha}^{-1/2} \sum_{x^n\in \TT_n(\alpha)} \ket{x^n}\, , \qquad \binom{n}{\alpha} \coloneqq \frac{n!}{\prod_{x=1}^D \alpha_x!}\, ,
\label{eq:Dicke-basis}
\ee
where $\TT_n(\alpha)$ denotes the set of all sequences $x^n\in \{1,\ldots,D\}^n$ with \deff{type} $\alpha$, i.e.\ in which every symbol $x\in \{1,\ldots, D\}$ appears exactly $\alpha_x$ times.

\section{Main results: technical statements and proofs} \label{sec:main_results}

We first prove a one-shot lower bound on the hypothesis-testing relative entropy in terms of the relative entropy of resource at the same blocklength.
We then derive the strengthened generalised quantum Stein's lemma, including its consequences for regularised quantum R\'enyi divergences.
Finally, we turn the one-shot bound into explicit sample-complexity estimates for asymmetric testing.
The diagram below records the manuscript-specific proof dependencies.
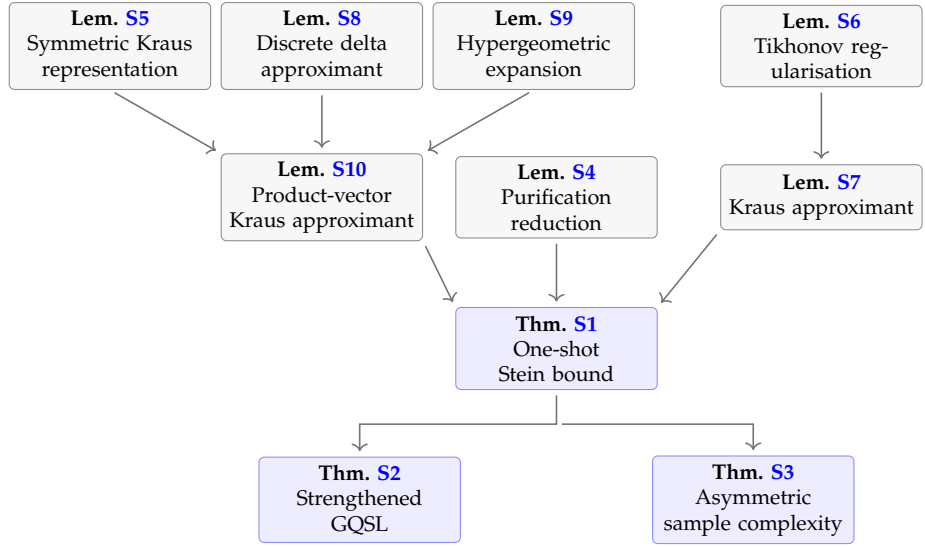
\begin{center}
\begin{tikzpicture}[
  box/.style={
    draw=black!45,
    rounded corners=2pt,
    align=center,
    text width=2.45cm,
    minimum height=0.86cm,
    inner sep=3pt,
    font=\scriptsize
  },
  lemma/.style={box, fill=black!3},
  result/.style={box, fill=blue!7, draw=blue!45},
  arr/.style={->, semithick, draw=black!55, shorten >=2pt, shorten <=2pt}
]
  \node[lemma] (symk) at (-5.1,6.0)
    {\textbf{Lem.~\ref{lem:clean-symmetric-kraus}}\\Symmetric Kraus representation};
  \node[lemma] (deltaapprox) at (-2.3,6.0)
    {\textbf{Lem.~\ref{lem:clean-discrete-delta-approximant}}\\Discrete polynomial localisation};
  \node[lemma] (hyper) at (0.5,6.0)
    {\textbf{Lem.~\ref{lem:clean-hypergeometric-expansion}}\\Hypergeometric expansion};
  \node[lemma] (tikh) at (4.3,6.0)
    {\textbf{Lem.~\ref{lem:Tikhonov_regularisation}}\\Tikhonov regularisation};

  \node[lemma] (productapprox) at (-2.3,4.0)
    {\textbf{Lem.~\ref{lem:clean-diagonal-kraus-approximant}}\\Product-vector Kraus approximant};
  \node[lemma] (purif) at (0.8,4.0)
    {\textbf{Lem.~\ref{lem:clean-purification-reduction}}\\Purification reduction};
  \node[lemma] (kapprox) at (4.3,4.0)
    {\textbf{Lem.~\ref{lem:clean-kraus-approximant}}\\Kraus approximant};

  \node[result] (oneshot) at (0.8,2.0)
    {\textbf{Thm.~\ref{thm:supp-one-shot-gqsl}}\\One-shot Stein bound};

  \node[result] (stein) at (-1.8,0.0)
    {\textbf{Thm.~\ref{thm:supp-limit-commutation}}\\Strengthened GQSL};
  \node[result] (sample) at (3.4,0.0)
    {\textbf{Thm.~\ref{thm:supp-asymmetric-sample-complexity}}\\ Asymmetric sample complexity};

  \draw[arr] (symk.south) -- (productapprox.north west);
  \draw[arr] (deltaapprox.south) -- (productapprox.north);
  \draw[arr] (hyper.south) -- (productapprox.north east);
  \draw[arr] (tikh.south) -- (kapprox.north);

  \draw[arr] (productapprox.south east) -- (oneshot.north west);
  \draw[arr] (purif.south) -- (oneshot.north);
  \draw[arr] (kapprox.south west) -- (oneshot.north east);

  \coordinate (split) at (0.8,1.0);
  \draw[arr] (oneshot.south) -- (split) -| (stein.north);
  \draw[arr] (split) -| (sample.north);
\end{tikzpicture}
\end{center}

\subsection{Hypothesis testing results}

We use the constant
\begin{equation}\label{eq:clean-K-constant}
  K_{d,\tau}
  \defeq
  16\of*{2\log d+\log(1/c_\tau)+1}.
\end{equation}
We also write $\kappa=\frac{15-2\pi}{32\sqrt{2}}$ for the universal constant in \cref{lem:clean-discrete-delta-approximant}.
We denote the binary entropy by
\begin{equation}
  h_2\of t
  \defeq
  -t\log t-\of{1-t}\log\of{1-t},
  \qquad
  0\le t\le1,
\end{equation}
with the convention $0\log0=0$.

\begin{theorem}[(One-shot formulation of the GQSL, extended version)] 
\label{thm:supp-one-shot-gqsl}
Let $\HH$ be a finite-dimensional Hilbert space, and let $\FF = (\FF_n)_{n\in \N}$ be a sequence of sets of free states $\FF_n\subseteq \D\of*{\HH^{\otimes n}}$ that satisfies the Brand\~ao--Plenio axioms~\eqref{ax:closed-convex}--\eqref{ax:permutation-invariant}. Let $\rho\in\D\of{\HH}$, $n\in\N^+$, and $\varepsilon\in(0,1)$ satisfy  $\of{d^2+2}^2 \le \log\frac1\varepsilon \le \frac{n}{240}$. Then,
\begin{equation}\label{eq:clean-one-shot-stein}
  \frac1n
  D_{\mathrm H}^{\varepsilon}
  \of{\rho^{\otimes n}\|\FF_n}
  \ge
  \frac1n
  D\of{\rho^{\otimes n}\|\FF_n}
  -
  f_{d,\tau}\of{\varepsilon,n},
\end{equation}
where
\begin{equation}\label{eq:clean-one-shot-penalty}
  f_{d,\tau}\of{\varepsilon,n}
  \defeq
  K_{d,\tau}
  \sqrt{\frac1n\log\frac1\varepsilon}
  +
  \sqrt{2\varepsilon}\of*{\log d+\log(1/c_\tau)}
  +
  \frac1n\of*{2+\log\frac1{1-2\varepsilon}}.
\end{equation}
\end{theorem}

This is the explicit version of \cref{thm:main-one-shot-gqsl} in the main text.
More precisely, the constants therein may be chosen as
\begin{equation*}
  \varepsilon_0
  =
  2^{-\of{d^2+2}^2},
  \qquad
  N
  =
  240\of{d^2+2}^2,
  \qquad
  A
  =
  \sqrt{2}\of*{\log d+\log(1/c_\tau)},
  \qquad
  B
  =
  K_{d,\tau}+1.
\end{equation*}
Indeed, if $n\ge N$ and $2^{-n/240}\le\varepsilon\le\varepsilon_0$, then $\of{d^2+2}^2\le\log(1/\varepsilon)\le n/240$, so \cref{thm:supp-one-shot-gqsl} applies.
Using $D_{\mathrm H}^{\varepsilon}=-\log\beta_{\varepsilon}$ and multiplying \eqref{eq:clean-one-shot-stein} by $n$, the second term in \eqref{eq:clean-one-shot-penalty} becomes $A n\sqrt{\varepsilon}$, while the first becomes $K_{d,\tau}\sqrt{n\log(1/\varepsilon)}$.
Moreover, $\varepsilon\le\varepsilon_0\le1/4$ gives $2+\log\frac1{1-2\varepsilon}\le3$, whereas $\sqrt{n\log(1/\varepsilon)}\ge d^2+2\ge3$.
Hence the remaining bounded term is at most $\sqrt{n\log(1/\varepsilon)}$ and can be combined with the first correction term, yielding \eqref{eq:one_shot_GQSL} with the stated value of $B$.

\begin{proof}
The case $d=1$ is trivial, so we henceforth assume that $d\ge2$. Define the type-I error exponent by
\begin{equation}
  s \defeq \frac1n\log\frac1\varepsilon.
\end{equation}
The assumed bounds on $\log(1/\varepsilon)$ imply
\begin{equation}
  \label{eq:bounds_s_n}
  0<s\le\frac1{240}, \qquad s n\ge\of{d^2+2}^2.
\end{equation}
Since $d\ge2$, we have $s n\ge 36$, and therefore $\varepsilon=2^{-s n}<1/2$.
The standard one-shot comparison~\cite[Theorem~12, Eq.~(96)]{RegulaLamiDatta} gives
\begin{equation}\label{eq:clean-DH-Dmax-comparison}
  D_{\max}^{\sqrt{1-\varepsilon}}
  \of{\rho^{\otimes n}\|\FF_n}
  +
  \log\frac1{1-\varepsilon}
  \le
  D_{\mathrm H}^{\varepsilon}
  \of{\rho^{\otimes n}\|\FF_n}.
\end{equation}
We first prove the smoothing estimate
\begin{equation}\label{eq:clean-one-shot-smooth-max-transfer}
  D_{\max}^{\sqrt{1-\varepsilon}}
  \of{\rho^{\otimes n}\|\FF_n}
  \ge
  D_{\max}^{\sqrt{2\varepsilon}}
  \of{\rho^{\otimes n}\|\FF_n}
  -
  K_{d,\tau}\sqrt{s}\,n
  -
  \log\frac1{1-2\varepsilon}.
\end{equation}
The minimum defining the smooth max-relative entropy is attained because the smoothing ball and $\FF_n$ are compact and the operator-order constraint is closed. Its value is finite because the free state $\tau^{\otimes n}$ has full rank.
Now, choose optimisers $\rho_n'$ and $\sigma_n\in\FF_n$ such that
\begin{equation} \label{eq:clean-one-shot-smoothed-order}
P\of{\rho^{\otimes n},\rho_n'} \le \sqrt{1-\varepsilon}, \qquad
\rho_n' \le 2^{D_{\max}^{\sqrt{1-\varepsilon}}\of{\rho^{\otimes n}\|\FF_n}} \sigma_n .
\end{equation}
Applying the permutation twirl to both states preserves the order relation, keeps $\sigma_n$ free, and cannot increase the purified distance from $\rho^{\otimes n}$. We may therefore assume that $\rho_n'$ and $\sigma_n$ are permutation invariant.

Now, let $\ket\theta\in\HH\otimes\mathcal R$ be a purification of $\rho$, where $\mathcal R\simeq\HH$. The first inequality in \eqref{eq:clean-one-shot-smoothed-order} implies that
\begin{equation}
F\of{\rho^{\otimes n},\rho_n'} \ge \sqrt\varepsilon = 2^{-s n/2}.
\end{equation}
Hence, \cref{lem:clean-purification-reduction} allows us to construct a symmetric purification $\ket\Psi\in\Sym^n\of{\HH\otimes\mathcal R}$ of $\rho_n'$ satisfying
\begin{equation}
\abs*{\braket{\theta^{\otimes n}}{\Psi}} \ge 2^{-s n/2}.
\end{equation}

We next convert this overlap into an operator inequality by blurring a suitable number of tensor factors. The one-site space $\HH\otimes\mathcal R$ has dimension $d^2$. Choose the number of blurred tensor factors as
\begin{equation}
k = \floor*{\sqrt{240\,s}\,n}.
\end{equation}
Let $\widetilde{\BB}_{n,k}$ be the maximally mixed blurring channel on $\of{\HH\otimes\mathcal R}^{\otimes n}$.
The polynomial supplied by the localisation technique in \cref{lem:clean-diagonal-kraus-approximant} has degree at most $\ceil*{\sqrt{2s\ln2/\kappa}\,n}$.
The inequality $s n\ge36$ verifies the condition $s n>1$ in \cref{lem:clean-diagonal-kraus-approximant}.
The bound \eqref{eq:clean-one-shot-sqrt-sn-bound}, together with the choices of $k$ and this degree, ensures that the latter is at most $k/2$.
Thus \cref{lem:clean-diagonal-kraus-approximant} applies with $\gamma=s$.
It yields a linear combination $K(c)$ of the Kraus operators such that
\begin{equation}
  K(c)\ket{\theta^{\otimes n}}
  =
  \ket{\theta^{\otimes n}},
  \qquad
  \norm*{K(c)-\proj{\theta^{\otimes n}}}_\infty^2
  \le
  2^{-2s n}.
\end{equation}
To apply \cref{lem:clean-kraus-approximant}, it remains to control the coefficient norm $\norm{c}_2$.
First, bounds in \eqref{eq:bounds_s_n} imply that
\begin{equation}
  \label{eq:clean-one-shot-sqrt-sn-bound}
  \sqrt{s}\,n = \frac{s n}{\sqrt{s}} \ge \sqrt{240}\of{d^2+2}^2.
\end{equation}
Then \eqref{eq:clean-one-shot-sqrt-sn-bound} implies
\begin{align}
  15\sqrt{s}\,n
  &\le
  k
  \le
  \sqrt{240s}\,n,
  &
  \ceil*{\sqrt{2s\ln2/\kappa}\,n}
  &\le
  4\sqrt{s}\,n,
  \notag\\
  k+1
  &\le
  17\sqrt{s}\,n,
  &
  k+d^2-2
  &\le
  17\sqrt{s}\,n.
\end{align}
Using $\log y\le\sqrt y$ for $y\ge16$ in \eqref{eq:clean-diagonal-kraus-approximant}, we obtain
\begin{align}
  \log\norm{c}_2^2
  &\le
  2k\log d
  +
  \log\of{k+1}
  +
  \of{d^2-2}\log\of{k+d^2-2}
  +
  \frac{2}{k\ln2}
  \ceil*{\sqrt{2s\ln2/\kappa}\,n}^{\!2}
  \notag\\
  &\le
  2\sqrt{240s}\,n\log d
  +
  2\sqrt{s}\,n
  +
  \frac{32}{15\ln2}\sqrt{s}\,n
  \notag\\
  &\le
  2\sqrt{240s}\,n\log d
  +
  8\sqrt{s}\,n.
  \label{eq:clean-one-shot-coefficient-bound}
\end{align}
The definition of $K_{d,\tau}$ now gives
\begin{align}
  &\log\norm{c}_2^2
  +
  \of*{
    s
    -
    K_{d,\tau}\sqrt{s}
    +
    \sqrt{240s}\log(1/c_\tau)
  }n
  \notag\\
  &\quad\le
  \Big[
    -\of{16-\sqrt{240}}
    \of{2\log d+\log(1/c_\tau)}
    -8
    +\sqrt{s}
  \Big]
  \sqrt{s}\,n
  \notag\\
  &\quad\le
  -1.
\end{align}
For the final inequality, $d\ge2$ and $s\le1/240$ make the expression in square brackets smaller than $-8$, while $s n\ge\of{d^2+2}^2$ implies $\sqrt{s}\,n>1$.
Equivalently,
\begin{equation}\label{eq:clean-one-shot-denominator}
  \norm{c}_2^2
  2^{\of*{
    s
    -
    K_{d,\tau}\sqrt{s}
    +
    \sqrt{240s}\log(1/c_\tau)
  }n}
  \le
  \frac12.
\end{equation}
This is precisely the denominator condition required by \cref{lem:clean-kraus-approximant}.
Applying that lemma with squared overlap $2^{-s n}$ and scale factor $2^{\of{K_{d,\tau}\sqrt{s}-\sqrt{240s}\log(1/c_\tau)}n}$ yields
\begin{align}
  \Tr\of*{
    \proj{\theta^{\otimes n}}
    -
    2^{\of*{
      K_{d,\tau}\sqrt{s}
      -
      \sqrt{240s}\log(1/c_\tau)
    }n}
    \widetilde{\BB}_{n,k}\of{\proj{\Psi}}
  }_+
  &\le
  \frac{2^{-2s n}}
       {
         2^{-s n}
         -
         2^{-\of*{
           K_{d,\tau}\sqrt{s}
           -
           \sqrt{240s}\log(1/c_\tau)
         }n}\norm{c}_2^2
       }\\
  &\le
  2^{1-s n}
  =
  2\varepsilon.
\end{align}
The second inequality follows directly from \eqref{eq:clean-one-shot-denominator}.
Although \cref{lem:clean-kraus-approximant} is stated for the compressed channel, the same estimate holds for the full blurring channel.
Indeed, the blurred output is permutation invariant, and therefore block diagonal with respect to the symmetric subspace and its orthogonal complement.
The first block is the compressed output, while the positive complementary block is subtracted inside the positive-part expression and can only decrease it.
Tracing out $\mathcal R^{\otimes n}$ and using data processing of $\Tr(\cdot)_+$ now gives
\begin{equation}
  \Tr\of*{
    \rho^{\otimes n}
    -
    2^{\of*{
      K_{d,\tau}\sqrt{s}
      -
      \sqrt{240s}\log(1/c_\tau)
    }n}
    \BB_{n,k}\of{\rho_n'}
  }_+
  \le
  2\varepsilon.
\end{equation}
We next replace maximally mixed blurring with $\tau$-blurring.
Since $k\le\sqrt{240s}\,n$ and $\tau\ge c_\tau\1/d$, \eqref{eq:tau-blurring-dominates-maximally-mixed} implies
\begin{equation}
  2^{\of*{
    K_{d,\tau}\sqrt{s}
    -
    \sqrt{240s}\log(1/c_\tau)
  }n}
  \BB_{n,k}\of{\rho_n'}
  \le
  2^{K_{d,\tau}\sqrt{s}\,n}
  \BB_{n,k}^{\tau}\of{\rho_n'}.
\end{equation}
The channel $\BB_{n,k}^{\tau}$ preserves the free set.
Combining this fact with \eqref{eq:clean-one-shot-smoothed-order} gives
\begin{equation}
  \Tr\of*{
    \rho^{\otimes n}
    -
    2^{D_{\max}^{\sqrt{1-\varepsilon}}\of{\rho^{\otimes n}\|\FF_n}+K_{d,\tau}\sqrt{s}\,n}
    \BB_{n,k}^{\tau}\of{\sigma_n}
  }_+
  \le
  2\varepsilon.
\end{equation}
The additive-defect form of the substate construction~\cite[Theorem~5]{RegulaLamiDatta} therefore provides a state at purified distance at most $\sqrt{2\varepsilon}$ from $\rho^{\otimes n}$ that is dominated by
\begin{equation}
  \frac{
    2^{D_{\max}^{\sqrt{1-\varepsilon}}\of{\rho^{\otimes n}\|\FF_n}+K_{d,\tau}\sqrt{s}\,n}
  }{1-2\varepsilon}
  \BB_{n,k}^{\tau}\of{\sigma_n}.
\end{equation}
Because $\BB_{n,k}^{\tau}\of{\sigma_n}$ is free by \eqref{eq:blurring-preserves-free}, this proves \eqref{eq:clean-one-shot-smooth-max-transfer}.

It remains to lower-bound $D_{\max}^{\sqrt{2\varepsilon}}$ by the relative entropy distance.
Let $\omega_n$ satisfy $P\of{\omega_n,\rho^{\otimes n}}\le\sqrt{2\varepsilon}$, and let $\sigma_n\in\FF_n$.
The Fuchs--van de Graaf inequality gives
\begin{equation}
  \frac12\norm*{\omega_n-\rho^{\otimes n}}_1
  \le
  \sqrt{2\varepsilon}.
\end{equation}
Since $\tau^{\otimes n}\in\FF_n$ and $\tau^{\otimes n}\ge c_\tau^n\1/d^n$, the standard mixing argument for the relative entropy distance~\cite[Lemma~7]{WinterContinuity} yields
\begin{equation}
  D_{\max}\of{\omega_n\|\sigma_n}
  \ge
  D\of{\omega_n\|\FF_n}
  \ge
  D\of{\rho^{\otimes n}\|\FF_n}
  -
  \sqrt{2\varepsilon}\,n\of*{\log d+\log(1/c_\tau)}
  -
  2.
\end{equation}
Here the entropy term in the continuity estimate is bounded by $2$.
Taking the infimum over $\omega_n$ and $\sigma_n$ gives
\begin{equation}
  D_{\max}^{\sqrt{2\varepsilon}}
  \of{\rho^{\otimes n}\|\FF_n}
  \ge
  D\of{\rho^{\otimes n}\|\FF_n}
  -
  \sqrt{2\varepsilon}\,n\of*{\log d+\log(1/c_\tau)}
  -
  2.
\end{equation}
Combining this inequality with \eqref{eq:clean-DH-Dmax-comparison} and \eqref{eq:clean-one-shot-smooth-max-transfer}, and discarding the nonnegative term $\log(1/(1-\varepsilon))$, yields
\begin{equation}
  D_{\mathrm H}^{\varepsilon}
  \of{\rho^{\otimes n}\|\FF_n}
  \ge
  D\of{\rho^{\otimes n}\|\FF_n}
  -
  K_{d,\tau}\sqrt{s}\,n
  -
  \sqrt{2\varepsilon}\,n\of*{\log d+\log(1/c_\tau)}
  -
  2
  -
  \log\frac1{1-2\varepsilon}.
\end{equation}
Substituting $s=n^{-1}\log(1/\varepsilon)$ and dividing by $n$ proves \eqref{eq:clean-one-shot-stein}.
\end{proof}

\begin{remark}[(Optimality of the square-root scaling)]
The square-root dependence on $s=n^{-1}\log(1/\varepsilon)$ in \cref{thm:supp-one-shot-gqsl} is optimal in general, although we do not claim that the constant $K_{d,\tau}$ is optimal.
Indeed, the Brand\~ao--Plenio axioms include the singleton i.i.d.\ family $\FF_n^\sigma=\set*{\sigma^{\otimes n}}$, where $\sigma$ is any full-rank state.
In this case the problem reduces to ordinary i.i.d.\ quantum hypothesis testing, whose exact error-exponent tradeoff is given, for $0<s<D\of{\sigma\|\rho}$, by the quantum Hoeffding bound~\cite{Hayashi2007Hoeffding,Nagaoka2006Converse,ANSV2008}:
\begin{equation*}
  \lim_{n\to\infty}
  \frac1n
  D_{\mathrm H}^{2^{-s n}}
  \of{\rho^{\otimes n}\|\sigma^{\otimes n}}
  =
  \sup_{\alpha\in(0,1)}
  \set*{
    D_\alpha\of{\rho\|\sigma}
    -
    \frac{\alpha}{1-\alpha}s
  }.
\end{equation*}
If the relative entropy variance $V\of{\rho\|\sigma}\defeq\Tr\rho\of*{\log\rho-\log\sigma}^2-D\of{\rho\|\sigma}^2$ is positive, expansion around $\alpha=1$ gives
\begin{equation*}
  \sup_{\alpha\in(0,1)}
  \set*{
    D_\alpha\of{\rho\|\sigma} - \frac{\alpha}{1-\alpha}s
  }
  = D\of{\rho\|\sigma} - \sqrt{2\ln2\,V\of{\rho\|\sigma}\,s} + O(s)
  \qquad
  \of{s\to0^+}.
\end{equation*}
Thus no bound valid uniformly over all resource theories satisfying the Brand\~ao--Plenio axioms can replace the $\sqrt{s}$ loss by $o\of{\sqrt{s}}$.
Equivalently, when $\varepsilon=2^{-s n}$, the correction of order $\sqrt{n\log(1/\varepsilon)}=n\sqrt{s}$ in \cref{thm:supp-one-shot-gqsl} has the optimal scaling.
\end{remark}

\begin{definition}
  For $\alpha\in(0,1)$, define the Petz R\'enyi divergence by
\begin{equation}
  D_{\alpha}\of{\omega\|\sigma}
  \defeq
  \frac1{\alpha-1}
  \log\Tr\of*{\omega^\alpha\sigma^{1-\alpha}}.
\end{equation}
For $\alpha\in(1/2,1)$, define the sandwiched R\'enyi divergence~\cite{MullerLennert2013,WildeWinterYang2014} by
\begin{equation}
  \widetilde{D}_{\alpha}\of{\omega\|\sigma}
  \defeq
  \frac1{\alpha-1}
  \log\Tr\of*{
    \of*{\sigma^{\frac{1-\alpha}{2\alpha}}\omega\sigma^{\frac{1-\alpha}{2\alpha}}}^{\alpha}
  }.
\end{equation}
Define
\begin{equation}
  D_{\alpha}^{\infty}\of{\rho\|\FF}
  \defeq
  \lim_{n\to\infty}
  \frac1n
  \inf_{\sigma_n\in\FF_n}
  D_{\alpha}\of{\rho^{\otimes n}\|\sigma_n},
\end{equation}
and define $\widetilde{D}_{\alpha}^{\infty}\of{\rho\|\FF}$ analogously.
These limits exist by subadditivity, which follows from the tensor-product axiom and the additivity of both R\'enyi divergences on tensor products.
\end{definition}

\begin{theorem}[(Solution to the limit commutation problem, extended version)]
\label{thm:supp-limit-commutation}
Assume \eqref{ax:closed-convex}--\eqref{ax:permutation-invariant}.
For every state $\rho\in\D\of{\HH}$ and every $0<s\le1/240$,
\begin{equation}\label{eq:clean-exponential-stein-quantitative}
  \liminf_{n\to\infty}
  \frac1n
  D_{\mathrm H}^{2^{-s n}}
  \of{\rho^{\otimes n}\|\FF_n}
  \ge
  D^\infty\of{\rho\|\FF}
  -
  K_{d,\tau}\sqrt{s}.
\end{equation}
Consequently,
\begin{equation}\label{eq:clean-exponential-stein}
  \lim_{s\to0^+}
  \liminf_{n\to\infty}
  \frac1n
  D_{\mathrm H}^{2^{-s n}}
  \of{\rho^{\otimes n}\|\FF_n}
  =
  \mathrm{Stein}\of{\rho\|\FF}
  =
  D^\infty\of{\rho\|\FF}.
\end{equation}
Moreover,
\begin{equation}
  \lim_{\alpha\to1^-}
  D_{\alpha}^{\infty}\of{\rho\|\FF}
  =
  \lim_{\alpha\to1^-}
  \widetilde{D}_{\alpha}^{\infty}\of{\rho\|\FF}
  =
  D^\infty\of{\rho\|\FF}.
\end{equation}
\end{theorem}

\begin{proof}
We first prove the quantitative lower bound \eqref{eq:clean-exponential-stein-quantitative}.
For fixed $0<s\le1/240$, \cref{thm:supp-one-shot-gqsl} applies with $\varepsilon=2^{-s n}$ whenever $n\ge\of{d^2+2}^2/s$.
For this choice of $\varepsilon$, the one-shot correction satisfies
\begin{equation}
  \lim_{n\to\infty}
  f_{d,\tau}\of{2^{-s n},n}
  =
  K_{d,\tau}\sqrt{s}.
\end{equation}
For every $n$, subadditivity and Fekete's lemma give
\begin{equation}
  \frac1nD\of{\rho^{\otimes n}\|\FF_n}
  \ge
  D^\infty\of{\rho\|\FF}.
\end{equation}
Taking the limit inferior in \cref{thm:supp-one-shot-gqsl} proves \eqref{eq:clean-exponential-stein-quantitative}.

We next identify the limiting hypothesis-testing exponent.
Fix $0<\varepsilon\le1/2$, fix $\sigma_n\in\FF_n$, and let $T$ satisfy $\Tr\rho^{\otimes n}T\ge1-\varepsilon$.
Binary data processing for the relative entropy gives
\begin{equation}
  D\of{\rho^{\otimes n}\|\sigma_n}
  \ge
  \of{1-\varepsilon}\log\frac1{\Tr\sigma_nT}
  -
  h_2\of{\varepsilon}.
\end{equation}
Optimising over $T$ and $\sigma_n$ yields
\begin{equation}\label{eq:generalized-stein-upper-bound-proof}
  D_{\mathrm H}^{\varepsilon}
  \of{\rho^{\otimes n}\|\FF_n}
  \le
  \frac{
    D\of{\rho^{\otimes n}\|\FF_n}
    +
    h_2\of{\varepsilon}
  }{1-\varepsilon}.
\end{equation}
Setting $\varepsilon=2^{-s n}$, dividing by $n$, and letting $n\to\infty$ gives
\begin{equation}
  \limsup_{n\to\infty}
  \frac1n
  D_{\mathrm H}^{2^{-s n}}
  \of{\rho^{\otimes n}\|\FF_n}
  \le
  D^\infty\of{\rho\|\FF}.
\end{equation}
Together with \eqref{eq:clean-exponential-stein-quantitative}, this upper bound implies
\begin{equation}
  \lim_{s\to0^+}
  \liminf_{n\to\infty}
  \frac1n
  D_{\mathrm H}^{2^{-s n}}
  \of{\rho^{\otimes n}\|\FF_n}
  =
  D^\infty\of{\rho\|\FF}.
\end{equation}
For a fixed $0<\varepsilon\le1/2$, \eqref{eq:generalized-stein-upper-bound-proof} similarly gives
\begin{equation}
  \liminf_{n\to\infty}
  \frac1n
  D_{\mathrm H}^{\varepsilon}
  \of{\rho^{\otimes n}\|\FF_n}
  \le
  \frac{D^\infty\of{\rho\|\FF}}{1-\varepsilon}.
\end{equation}
Letting $\varepsilon\to0^+$ proves $\mathrm{Stein}\of{\rho\|\FF}\le D^\infty\of{\rho\|\FF}$.
Conversely, for every fixed $\varepsilon>0$ and $0<s\le1/240$, one has $2^{-s n}\le\varepsilon$ for all sufficiently large $n$.
Monotonicity of the hypothesis-testing divergence and \eqref{eq:clean-exponential-stein-quantitative} therefore imply
\begin{equation}
  \liminf_{n\to\infty}
  \frac1n
  D_{\mathrm H}^{\varepsilon}
  \of{\rho^{\otimes n}\|\FF_n}
  \ge
  D^\infty\of{\rho\|\FF}
  -
  K_{d,\tau}\sqrt{s}.
\end{equation}
Letting $s\to0^+$ gives the reverse inequality and completes the proof of \eqref{eq:clean-exponential-stein}.

It remains to prove the R\'enyi limits. For $0<\alpha<1$ in the Petz case and $1/2<\alpha<1$ in the sandwiched case, monotonicity in the R\'enyi parameter gives
\begin{equation}
  D_{\alpha}\of{\omega\|\sigma}
  \le
  D\of{\omega\|\sigma},
  \qquad
  \widetilde{D}_{\alpha}\of{\omega\|\sigma}
  \le
  D\of{\omega\|\sigma}.
\end{equation}
Minimising over $\sigma_n\in\FF_n$, dividing by $n$, and taking the limit inferior gives
\begin{equation}
  D_{\alpha}^{\infty}\of{\rho\|\FF}
  \le
  D^\infty\of{\rho\|\FF},
  \qquad
  \widetilde{D}_{\alpha}^{\infty}\of{\rho\|\FF}
  \le
  D^\infty\of{\rho\|\FF}.
\end{equation}
If $D^\infty\of{\rho\|\FF}=0$, these upper bounds and nonnegativity already prove both R\'enyi limits. Assume therefore that $D^\infty\of{\rho\|\FF}>0$, fix $0<r<D^\infty\of{\rho\|\FF}$, and choose $0<s\le1/240$ so small that
\begin{equation}
  r
  <
  D^\infty\of{\rho\|\FF}
  -
  K_{d,\tau}\sqrt{s}.
\end{equation}
By \eqref{eq:clean-exponential-stein-quantitative}, the strict inequality above ensures that $D_{\mathrm H}^{2^{-s n}}\of{\rho^{\otimes n}\|\FF_n}>r n$ for all sufficiently large $n$.
Hence there is a test $T_n$ satisfying
\begin{equation}\label{eq:clean-renyi-test-errors}
  \alpha_n\of{T_n}
  \le
  2^{-s n},
  \qquad
  \beta_n\of{T_n}
  \le
  2^{-r n}.
\end{equation}
Fix $\sigma_n\in\FF_n$ and apply data processing with respect to the binary measurement $\set*{T_n,\1-T_n}$.
For the Petz divergence, this gives
\begin{align}
  D_{\alpha}\of{\rho^{\otimes n}\|\sigma_n}
  &\ge
  \frac1{\alpha-1}
  \log\of*{
    \of*{1-\alpha_n\of{T_n}}^{\alpha}
    \of*{\Tr\sigma_nT_n}^{1-\alpha}
    +
    \alpha_n\of{T_n}^{\alpha}
    \of*{1-\Tr\sigma_nT_n}^{1-\alpha}
  }\\
  &\ge
  -\frac1{1-\alpha}
  \log\of*{
    2^{-(1-\alpha)r n}
    +
    2^{-\alpha s n}
  }.
\end{align}
For $\alpha>1/2$, data processing of the sandwiched R\'enyi divergence~\cite[Theorem~1]{FrankLieb2013} gives the same bound for $\widetilde{D}_{\alpha}\of{\rho^{\otimes n}\|\sigma_n}$, since it agrees with the Petz divergence on classical distributions.
Dividing by $n$, taking the infimum over $\sigma_n$, and then taking the limit inferior in $n$ yields
\begin{equation}
  D_{\alpha}^{\infty}\of{\rho\|\FF}
  \ge
  \min\set*{r,\frac{\alpha s}{1-\alpha}},
  \qquad
  \widetilde{D}_{\alpha}^{\infty}\of{\rho\|\FF}
  \ge
  \min\set*{r,\frac{\alpha s}{1-\alpha}}.
\end{equation}
Letting $\alpha\to1^-$ gives a lower bound of $r$ for both regularised divergences.
Since $r$ can be chosen arbitrarily close to $D^\infty\of{\rho\|\FF}$, we obtain
\begin{equation}
  \liminf_{\alpha\to1^-}
  D_{\alpha}^{\infty}\of{\rho\|\FF}
  \ge
  D^\infty\of{\rho\|\FF},
  \qquad
  \liminf_{\alpha\to1^-}
  \widetilde{D}_{\alpha}^{\infty}\of{\rho\|\FF}
  \ge
  D^\infty\of{\rho\|\FF}.
\end{equation}
The matching upper bounds established above prove the asserted limits of the regularised R\'enyi divergences.

For completeness, we also justify the second iterated limit in \eqref{eq:limit_commutation_problem}.
For every fixed $n$, monotonicity in $\alpha$ immediately gives
\begin{equation}
  \lim_{\alpha\to1^-}
  \inf_{\sigma_n\in\FF_n}
  D_\alpha\of{\rho^{\otimes n}\|\sigma_n}
  \le
  D\of{\rho^{\otimes n}\|\FF_n}.
\end{equation}
For the reverse inequality, choose minimisers along a sequence $\alpha_j\to1^-$ and, by compactness of $\FF_n$, pass to a subsequence converging to some $\sigma_n^\star\in\FF_n$.
For every fixed $\alpha'<1$, monotonicity in the R\'enyi parameter and continuity in the second argument give
\begin{equation}
  D_{\alpha'}\of{\rho^{\otimes n}\|\sigma_n^\star}
  \le
  \lim_{j\to\infty}
  \inf_{\sigma_n\in\FF_n}
  D_{\alpha_j}\of{\rho^{\otimes n}\|\sigma_n}.
\end{equation}
Letting $\alpha'\to1^-$ shows that the reverse inequality holds, and hence
\begin{equation}
  \lim_{\alpha\to1^-}
  \inf_{\sigma_n\in\FF_n}
  D_\alpha\of{\rho^{\otimes n}\|\sigma_n}
  =
  D\of{\rho^{\otimes n}\|\FF_n}.
\end{equation}
Taking $n\to\infty$ now gives the right-hand side of \eqref{eq:limit_commutation_problem} and completes the proof.
\end{proof}

\subsection{Sample complexity of resource testing}

\begin{theorem}[(Sample complexity of asymmetric resource testing)]
\label{thm:supp-asymmetric-sample-complexity}
Assume \eqref{ax:closed-convex}--\eqref{ax:permutation-invariant}.
Let $\rho\in\D\of{\HH}$ satisfy $D^\infty\of{\rho\|\FF}>0$.
Fix an allowed type-I error $\varepsilon\in(0,1)$.
For fixed $\rho$ and $\varepsilon$, as $\delta\to0$, the sample complexity defined in \eqref{eq:sample-complexity} satisfies
\begin{equation}
  N_{\e,\delta}\of{\rho\|\FF}
  = \Theta\of*{
    \frac{\log(1/\delta)}
         {D^\infty\of{\rho\|\FF}}
  },
\end{equation}
meaning that $N_{\e,\delta}\of{\rho\|\FF}$ is upper and lower bounded by constant multiples of $\frac{\log(1/\delta)}{D^\infty\of{\rho\|\FF}}$ as $\delta \to 0^+$; the involved constants only depend on the underlying resource theory and on $\e$, and not on the state $\rho$.
\end{theorem}

\begin{proof}
For the upper bound, choose a fixed $\tilde\varepsilon\in(0,\varepsilon]$ such that
\begin{equation}
  \log\frac1{\tilde\varepsilon}
  \ge
  \of{d^2+2}^2,
  \qquad
  \sqrt{2\tilde\varepsilon}
  \of*{\log d+\log(1/c_\tau)}
  \le
  \frac14D^\infty\of{\rho\|\FF}.
\end{equation}
Choose $n$ so that
\begin{align}
  n
  &\ge
  \frac{4K_{d,\tau}^2}{D^\infty\of{\rho\|\FF}^2}
  \log\frac1{\tilde\varepsilon},
  \label{eq:clean-asymmetric-type-I-condition}\\
  n
  &\ge
  \frac4{D^\infty\of{\rho\|\FF}}
  \of*{3+\log\frac1\delta}.
  \label{eq:clean-asymmetric-type-II-condition}
\end{align}
The full-rank free state gives
\begin{equation}
  D^\infty\of{\rho\|\FF}
  \le
  D\of{\rho\|\tau}
  \le
  \log d+\log(1/c_\tau)
  \le
  K_{d,\tau}/16.
\end{equation}
The condition \eqref{eq:clean-asymmetric-type-I-condition} and the previous bound $D^\infty\of{\rho\|\FF}\le K_{d,\tau}/16$ imply
\begin{equation}
  \log\frac1{\tilde\varepsilon}
  \le
  \frac n{240},
  \qquad
  K_{d,\tau}
  \sqrt{\frac1n\log\frac1{\tilde\varepsilon}}
  \le
  \frac12D^\infty\of{\rho\|\FF}.
\end{equation}
So our choice of $\tilde\varepsilon$ satisfies conditions in \cref{thm:supp-one-shot-gqsl} and also ensures
\begin{equation}
  \sqrt{2\tilde\varepsilon}
  \of*{\log d+\log(1/c_\tau)}
  \le
  \frac14D^\infty\of{\rho\|\FF},
  \qquad
  2+\log\frac1{1-2\tilde\varepsilon}
  \le
  3.
\end{equation}
Applying \cref{thm:supp-one-shot-gqsl} and using $D\of{\rho^{\otimes n}\|\FF_n}\ge nD^\infty\of{\rho\|\FF}$ yields
\begin{align}
  D_{\mathrm H}^{\tilde\varepsilon}
  \of{\rho^{\otimes n}\|\FF_n}
  &\ge
  nD^\infty\of{\rho\|\FF} - \frac n2D^\infty\of{\rho\|\FF} - \frac n4D^\infty\of{\rho\|\FF} - 3 \notag\\
  &= 
  \frac n4D^\infty\of{\rho\|\FF} - 3.
\end{align}
By \eqref{eq:clean-asymmetric-type-II-condition}, the last expression is at least $\log(1/\delta)$.
The definition of the hypothesis-testing relative entropy therefore provides a test whose type-I and type-II errors are at most $\tilde\varepsilon\le\varepsilon$ and $\delta$, respectively.
For fixed $\rho$ and $\varepsilon$, \eqref{eq:clean-asymmetric-type-I-condition} is independent of $\delta$, while \eqref{eq:clean-asymmetric-type-II-condition} gives 
\begin{equation}
  N_{\e,\delta}\of{\rho\|\FF}
  =
  O\of*{ \frac{\log(1/\delta)}{D^\infty\of{\rho\|\FF}}
  }.
\end{equation}
For the lower bound, let $T_n$ be any test satisfying $\alpha_n\of{T_n}\le\varepsilon$ and $\beta_n\of{T_n}\le\delta$.
For every $\sigma_n\in\FF_n$, binary data processing and $h_2\le1$ give
\begin{align}
  D\of{\rho^{\otimes n}\|\sigma_n}
  &\ge
  \of*{1-\alpha_n\of{T_n}}
  \log\frac1{\Tr\sigma_nT_n}
  -
  h_2\of*{\alpha_n\of{T_n}}
  \notag\\
  &\ge
  \of{1-\varepsilon}\log\frac1\delta
  -
  1.
\end{align}
Taking the infimum over $\sigma_n\in\FF_n$ yields
\begin{equation}\label{eq:clean-asymmetric-sample-lower-relative-entropy}
  D\of{\rho^{\otimes n}\|\FF_n}
  \ge
  \of{1-\varepsilon}\log\frac1\delta
  -
  1.
\end{equation}
Moreover, the full-rank free state implies
\begin{equation}
  \beta_n\of{T_n}
  \ge
  \Tr\tau^{\otimes n}T_n
  \ge
  \of*{\frac{c_\tau}{d}}^n
  \Tr T_n
  \ge
  \of*{\frac{c_\tau}{d}}^n
  \Tr\rho^{\otimes n}T_n
  \ge
  \of{1-\varepsilon}
  \of*{\frac{c_\tau}{d}}^n.
\end{equation}
Consequently, every admissible copy number tends to infinity as $\delta\to0$.
Since $n^{-1}D\of{\rho^{\otimes n}\|\FF_n}\to D^\infty\of{\rho\|\FF}>0$, for all sufficiently large $n$ we have
\begin{equation}
  D\of{\rho^{\otimes n}\|\FF_n}
  \le
  2nD^\infty\of{\rho\|\FF}.
\end{equation}
Combining this with \eqref{eq:clean-asymmetric-sample-lower-relative-entropy} proves
\begin{equation}
  N_{\e,\delta}\of{\rho\|\FF}
  =
  \Omega\of*{
    \frac{\log(1/\delta)}
         {D^\infty\of{\rho\|\FF}}
  }.
\end{equation}
\end{proof}

\section{Technical lemmas}

\begin{lemma}[{\cite[Lemma~III.4]{BrandaoPlenio}}]\label{lem:clean-purification-reduction}
Let $\omega\in\D\of{\HH}$ have a purification $\ket\theta\in\HH\otimes\HH'$, where $\HH'\simeq\HH$.
Let $\rho_n\in\D\of{\HH^{\otimes n}}$ be permutation invariant.
Then $\rho_n$ has a permutation-symmetric purification $\ket{\Psi_n}\in\Sym^n\of{\HH\otimes\HH'}$ such that
\begin{equation}
  \braket{\theta^{\otimes n}}{\Psi_n}
  =
  F\of{\omega^{\otimes n},\rho_n}.
\end{equation}
\end{lemma}

We now consider the action of the blurring channel $\BB_{n,k}$, defined by~\eqref{eq:blurring-channel} with the convention~\eqref{eq:blurring-with-identity}, as restricted to the symmetric subspace $\Sym^n(\HH)$. To wit, we will define the \deff{compressed blurring channel} by~\cite[Eq.~(90)]{LamiGQSL}
\bb
X \longmapsto \Pi_n\BB_{n,k}\big( \Pi_n^\dag X \Pi_n\big)\Pi_n^\dag ,
\label{eq:compressed-blurring}
\ee
where 
\bb
\Pi_n :\HH^{\otimes n} \to \Sym^n(\HH)
\label{eq:Pi_n}
\ee
is the projector onto the symmetric subspace, and $X$ is an arbitrary operator on $\Sym^n(\HH)$. The following generalises~\cite[Lemma~17]{LamiGQSL}.

\begin{lemma}[(Symmetric-subspace Kraus representation)]\label{lem:clean-symmetric-kraus}
Let $\HH \simeq\C^D$, let $1\le k\le n$, and let $\Pi_n$ be the projection onto $\Sym^n\of{\HH}$. For $\alpha\in\N^D$ with $\abs{\alpha} = \sum_{x=1}^D \alpha_x = n$, let $\ket{\alpha}$ be the normalised Dicke vector of weight $\alpha$ given by~\eqref{eq:Dicke-basis}. 
Then the compressed blurring channel~\eqref{eq:compressed-blurring} has Kraus operators $K_{\beta',\beta}$, indexed by $\beta,\beta'\in\N^D$ with $\abs{\beta}=\abs{\beta'}=k$ and satisfying
\bb \label{eq:clean-symmetric-kraus}
K_{\beta',\beta}\ket{\alpha} =
  \begin{cases}
    \displaystyle
    \sof*{
      \frac{\binom k\beta\binom k{\beta'}}{D^k}
      \frac{\binom{n-k}{\alpha-\beta}^2}
           {\binom n\alpha\binom n{\alpha-\beta+\beta'}}
    }^{1/2}
    \ket{\alpha-\beta+\beta'}
      &\quad \text{if $\beta\le\alpha$,} \\[3ex]
    0,
      &\quad \text{otherwise,}
  \end{cases}
\ee
where the inequality $\beta\le\alpha$ is understood to be entry-wise. In particular, the diagonal Kraus operators act by
\bb
\label{eq:clean-diagonal-symmetric-kraus}
K_{\beta,\beta}\ket{\alpha} = D^{-k/2} \frac{\binom k\beta\binom{n-k}{\alpha-\beta}}{\binom n\alpha}\ket{\alpha}\, .
\ee
\end{lemma}

\begin{proof}
Due to permutation symmetry, all the terms appearing in~\eqref{eq:blurring-channel} act in the same way, i.e.
\bb
\Pi_n \left[ \left(I_{J^c} \otimes \big(\omega^{\otimes k} \Tr\big)_J\right) \big( \Pi_n^\dag X \Pi_n\big) \right] \Pi_n^\dag
\ee
does not depend on $J$, for a fixed $k = |J|$. We can therefore, without loss of generality, pick $J = \{1,\ldots,k\}$. The standard-basis Kraus operators of the corresponding term in~\eqref{eq:blurring-channel} are thus
\begin{equation}
D^{-k/2}\, \ketbra{y^k}{z^k} \otimes \1^{\otimes(n-k)}\, ,
\end{equation}
indexed by $y^k,z^k\in \{1,\ldots,D\}^k$. 
Let $\beta',\beta \in \N^D$ denote the types of $y^k,z^k$, respectively. If $\beta\not\le\alpha$, the Kraus operator annihilates $\ket\alpha$.
If $\beta\le\alpha$, put $\gamma=\alpha-\beta+\beta'$.
Using the Dicke expansion,
\bb
\bra{\gamma} \left( D^{-k/2}\, \ketbra{y^k}{z^k} \otimes \1^{\otimes(n-k)} \right) \ket{\alpha} &= D^{-k/2} \frac{\binom{n-k}{\alpha-\beta}}{\sqrt{\binom n\alpha\binom n\gamma}}\, .
\ee
For fixed $\beta',\beta$, there are $\binom k{\beta'}\binom k\beta$ Kraus operators with these two types. Coalescing identical compressed Kraus operators multiplies the preceding amplitude by the square root of this number, giving~\eqref{eq:clean-symmetric-kraus}. Setting $\beta'=\beta$ then yields \eqref{eq:clean-diagonal-symmetric-kraus}.
\end{proof}

We record the standard variational identity for Tikhonov regularisation in the form needed below; see~\cite[Lemma~14, Eq.~(13)]{MicchelliPontil2005} for the equivalent regularised least-squares formula.

\begin{lemma}[(Variational form of Tikhonov regularisation)] \label{lem:Tikhonov_regularisation}
Let $W\in \C^{m\times n}$ be an $m\times n$ complex matrix, $\ket{v} \in \C^m$ a complex vector, and $\lambda>0$. Then
\bb
\min_{\ket{z}\in \C^n} \Big\{ \norm{\ket{z}}^2 + \frac1\lambda \norm{\ket{v} - W\ket{z}}^2 \Big\} = \bra{v}\big(WW^\dag + \lambda\1_m \big)^{-1} \ket{v}\, .
\label{eq:Tikhonov_regularisation}
\ee
\end{lemma}

\begin{proof}
It suffices to complete the square:
\begin{align}
&\norm{\ket{z}}^2 + \frac1\lambda \norm{\ket{v} - W\ket{z}}^2 \nonumber \\
&\qquad = \bra{z} \big(\1_n + \tfrac1\lambda W^\dag W \big) \ket{z} - \frac2\lambda \Re \bra{v} W \ket{z} + \frac1\lambda \norm{\ket{v}}^2 \nonumber \\
&\qquad = \left\| \big(\1_n + \tfrac1\lambda W^\dag W \big)^{1/2} \ket{z}\right\|^2 - \frac2\lambda \Re \bra{v} W \ket{z} + \frac1\lambda \norm{\ket{v}}^2 \nonumber \\
&\qquad = \left\| \big(\1_n + \tfrac1\lambda W^\dag W \big)^{1/2} \ket{z} - \frac1\lambda \big(\1_n + \tfrac1\lambda W^\dag W \big)^{-1/2} W^\dag \ket{v} \right\|^2 \\
&\qquad \quad - \frac{1}{\lambda^2} \bra{v} W \big(\1_n + \tfrac1\lambda W^\dag W \big)^{-1} W^\dag \ket{v} + \frac1\lambda \braket{v}{v} \nonumber \\
&\qquad = \left\| \big(\1_n + \tfrac1\lambda W^\dag W \big)^{1/2} \ket{z} - \frac1\lambda \big(\1_n + \tfrac1\lambda W^\dag W \big)^{-1/2} W^\dag \ket{v} \right\|^2 + \bra{v} \big(WW^\dag + \lambda\1_m \big)^{-1} \ket{v}\, , \nonumber 
\end{align}
where in the last step we used the Woodbury matrix identity. Minimising over $\ket{z}$ annihilates the first term, completing the proof. 
\end{proof}

\begin{lemma}[(Kraus approximant)]\label{lem:clean-kraus-approximant}
For some $M>0$ and $\varepsilon\in(0,1]$, let $\ket{v},\ket{\Psi}\in \HH$ be two unit vectors in the Hilbert space $\HH$ such that 
\bb
\abs*{\braket{v}{\Psi}}^2\ge\varepsilon\, .
\ee
Let $(K_x)_{x\in \XX}$ be a finite collection of operators. Then, for all vectors $c\in \C^{\XX}$ with the property that $\norm{c}^2 < M\e$, it holds that
\bb
\label{eq:clean-kraus-approximant}
\Tr\of*{\proj{v}-M\sum\nolimits_x K_x\proj{\Psi}K_x^\dag}_+ \le \frac{\norm*{K(c)-\proj{v}}_\infty^2}{\varepsilon-\norm{c}^2/M}\, ,
\ee
where $K(c) \defeq \sum_x c_x K_x$, and $\|\cdot\|_\infty$ denotes the operator norm.
\end{lemma}

\begin{proof}
Set $A_M \defeq M\sum\nolimits_x K_x\proj{\Psi}K_x^\dag$, and define $W:\C^{\XX}\to\HH$ by
\begin{equation}
  W
  \defeq
  M^{1/2}\sum_{x\in\XX}K_x\ketbra{\Psi}{x},
  \qquad
  WW^\dag=A_M.
\end{equation}
Due to the Weyl monotonicity principle, the operator $\proj{v}-A_M$ has at most a single positive eigenvalue.
Therefore,
\bb
\Tr\of{\proj{v}-A_M}_+ = \inf\Set{ \lambda > 0  \given \proj{v}-A_M \leq \lambda \1} = \inf\Set{ \lambda > 0  \given \bra{v}(A_M + \lambda \1)^{-1} \ket{v} \leq 1}\, .
\label{clean-kraus-approximant_proof_eq1}
\ee 
Setting $\alpha=\braket{v}{\Psi}$, we now write
\bb
\bra{v}(A_M + \lambda \1)^{-1} \ket{v} &\eqt{(i)} \min_{\ket{z}\in \C^{\XX}} \Big\{ \norm{\ket{z}}^2 + \frac1\lambda \norm{\ket{v} - W\ket{z}}^2 \Big\} \\
&\leqt{(ii)} \frac{1}{M|\alpha|^2} \norm{c}^2 + \frac1\lambda \norm*{\ket{v} - \tfrac1\alpha K(c) \ket{\Psi}}^2 \\
&\leqt{(iii)} \frac{1}{M|\alpha|^2} \norm{c}^2 + \frac{1}{\lambda |\alpha|^2} \norm*{\proj{v} - K(c)}_\infty^2 \\
&\leqt{(iv)} \frac{1}{M\e} \norm{c}^2 + \frac{1}{\lambda \e} \norm*{\proj{v} - K(c)}_\infty^2
\label{clean-kraus-approximant_proof_eq2}
\ee
Here, in~(i) we applied Lemma~\ref{lem:Tikhonov_regularisation} to the operator $W$ defined above; in~(ii) we substituted the ansatz $\ket{z} \to M^{-1/2} \alpha^{-1} c$; in~(iii) we realised that
\bb
\ket{v} - \alpha^{-1} K(c) \ket{\Psi} = \alpha^{-1} \left( \proj{v} - K(c) \right) \ket{\Psi}\, ,
\ee
so that 
\bb
\norm*{\ket{v} - \alpha^{-1} K(c) \ket{\Psi}} \leq |\alpha|^{-1} \norm*{\proj{v} - K(c)}_\infty\, ;
\ee
finally, (iv)~holds because $|\alpha|^2 = \abs{\braket{v}{\Psi}}^2 \ge \e$ by assumption.

Now, to upper bound~\eqref{clean-kraus-approximant_proof_eq1} we can find the minimum $\lambda>0$ such that the rightmost side of~\eqref{clean-kraus-approximant_proof_eq2} is equal to $1$. Since this is equal to
\bb
\lambda_0 \defeq \frac{\norm*{K(c)-\proj{v}}_\infty^2}{\varepsilon-\norm{c}^2/M}\, ,
\ee
we obtain that
\bb
\Tr\of{\proj{v}-A_M}_+ \leq \lambda_0\, ,
\ee
which concludes the proof.
\end{proof}

The next lemma implements a localisation technique by approximating the Kronecker delta function at $s=0$ with a low-degree polynomial $q_r$.
In terms of the Hamming weight $s$, the polynomial $1-q_r(s)$ approximates the OR function on $n$ bits, with the exact normalisation $q_r(0)=1$.
The optimal $\varepsilon$-approximate degree of OR is $\Theta\big(\sqrt{n\log(1/\varepsilon)}\big)$ for $2^{-n}\le\varepsilon\le1/3$; this tradeoff was obtained up to logarithmic factors by Kahn, Linial and Samorodnitsky~\cite{KahnLinialSamorodnitsky1996} and sharply by Buhrman, Cleve, de Wolf and Zalka~\cite{BuhrmanCleveDeWolfZalka1999}; see also~\cite{deWolf2008,Sherstov2009}.
We include a self-contained proof with an explicit constant, following the bounded Chebyshev construction of Sherstov~\cite[Lemma~4.4 and Theorem~4.5]{Sherstov2020}.

\input{fig_delta_approximant.tex}

\begin{lemma}[(Discrete polynomial localisation)]\label{lem:clean-discrete-delta-approximant}
For every $n\ge1$ and every integer $0\le r\le n$, there is a real polynomial $q_r$ of degree at most $r$ such that
\begin{equation}
  q_r\of0=1
\end{equation}
and
\begin{equation}\label{eq:clean-discrete-delta-approximant}
\max_{s\in \{1,\ldots,n\}} \abs*{q_r\of s} \le 2 e^{-\kappa r^2/n}\, ,
\end{equation}
where $\kappa$ is a universal constant, and we can take for instance $\kappa=\frac{15-2\pi}{32\sqrt{2}} \approx 0.1926$.
\end{lemma}

\begin{proof}
For $r=0$, take $q_0\equiv1$.
The remaining claim is obvious for $n=1$, so we can assume that $n\ge2$ and $r\ge1$.
Let $T_m$ be the Chebyshev polynomial of the first kind of degree $m$, defined by the identity $T_m(\cos\theta)=\cos(m\theta)$.
In particular, the first non-negative zero of $T_m$ is at $x=\cos\frac{\pi}{2m}$, and
\bb
\max_{x\in [-1,1]} \abs*{T_m(x)} = 1\, .
\label{eq:Chebyshev-max-unit-interval}
\ee
For $x\ge  1$, one can also write the explicit formula
\bb
T_m(x) = \cosh\big( m \arccosh(x)\big)\, ,
\label{eq:Chebyshev-cosh-arccosh}
\ee
from which we immediately observe that $T_m(x)>0$ for all $x\ge  1$.

Now, fix some integers $L\in \{1,\ldots, n-1\}$ and $m\in \{1,\ldots,r\}$, to be determined later. The main idea of the proof is to construct $q_r$ by multiplying two different polynomials:
\begin{itemize}
\item First, a polynomial $B_m$ of degree $m$, normalised at $0$, i.e.\ such that $B_m(0)=1$, and with small modulus on the whole real interval $[L,n]$. The construction of $B_m$ rests on classic results in approximation theory, and, indeed, $B_m$ will be nothing but a rescaled and translated version of $T_m$, where the rescaling and the translation are meant to map the interval $[-1,1]$ to $[L,n]$.
\item The problem with Chebyshev polynomials is that they have, in general, a large modulus outside of the interval where they are designed to be small. (For $T_m$, this interval is $[-1,1]$; for our $B_m$, it will be instead $[L,n]$.) To remedy this problem, we will use the remaining degree budget $r-m$ to multiply $B_m$ by appropriate Chebyshev polynomials that are designed to vanish for every integer $s\in [1,L)$ (if any), and are anyway small in modulus on the whole interval $[0,n]$.
\end{itemize}

We now execute this plan. Start by setting
\bb
B_m(x) \defeq \frac{T_m\of*{\frac{n-x}{n-L}}}{T_m\of*{\frac{n}{n-L}}}\, ,
\ee
so that $B_m(0)=1$, and 
\bb
\max_{x\in [L,n]} \abs*{B_m(x)} &= \frac{1}{T_m\of*{\frac{n}{n-L}}} \max_{x\in [L,n]} \abs*{T_m\of*{\frac{n-x}{n-L}}} \\
&\eqt{(i)} \frac{1}{T_m\of*{\frac{n}{n-L}}} \max_{y\in [0,1]} \abs*{T_m\of*{y}} \\
&\eqt{(ii)} \frac{1}{T_m\of*{\frac{n}{n-L}}} \\
&\leqt{(iii)} 2 \exp\left[ -m \arccosh\of*{\frac{n}{n-L}} \right]\, .
\label{eq:B_m-bound}
\ee
Here, in~(i) we changed variables by setting $y \defeq \frac{n-x}{n-L}$, in~(ii) we leveraged~\eqref{eq:Chebyshev-max-unit-interval}, and in~(iii) we exploited~\eqref{eq:Chebyshev-cosh-arccosh}.

For any integer $j\in (0,L)$, if any, we now set
\bb
m_j=\ceil*{\frac\pi4\sqrt{\frac{n}{j}}}, \qquad a_j=\frac nj\of*{1-\cos\frac{\pi}{2m_j}}, \qquad Z_j(x) = T_{m_j}\of*{1-a_j\frac{x}{n}}\, .
\label{eq:polynomials_T_mj}
\ee
Since $0\leq 1-\cos t\le t^2/2$, one has $0\leq a_j\le2$; therefore, we have $1-a_j\frac{x}{n}\in [-1,1]$ for all $x\in [0,n]$, entailing, via~\eqref{eq:Chebyshev-max-unit-interval}, that 
\bb
\abs{Z_j(x)}\le 1\qquad \forall\ 0\le x\le n\, .
\label{eq:Z_j-bound-0-n}
\ee
The polynomials $Z_j$ are normalised at $0$, i.e.\ $Z_j(0)=1$, and moreover
\bb
Z_j(j) = T_{m_j}\of*{\cos\frac{\pi}{2m_j}} = \cos\frac{\pi}{2} = 0 \qquad \forall\ j\in (0,L)\cap \N\, .
\label{eq:Z_j-vanishes-at-integers}
\ee

We are finally ready to define
\bb
q_r(x) = B_m(x) \prod_{j\in (0,L)\cap \N} Z_j(x)\, ,
\ee
where it is understood that $q_r(x)=B_m(x)$ if there are no integers in the interval $(0,L)$, i.e.\ if $L=1$. Clearly, $q_r(0)=1$, and moreover
\bb
\max_{s\in \{1,\ldots,n\}} \abs*{q_r(s)}\ &\eqt{(iv)}\ \max_{s\in \{L,\ldots,n\}} \abs*{q_r(s)} \\
&\leqt{(v)}\ 2 \exp\left[ -m \arccosh\of*{\frac{n}{n-L}} \right] ,
\label{eq:first-bound-q_r}
\ee
where (iv)~holds due to~\eqref{eq:Z_j-vanishes-at-integers}, and~(v) due to~\eqref{eq:B_m-bound} and~\eqref{eq:Z_j-bound-0-n}.

We will now make the ansatzes
\bb
L \defeq \ceil*{\alpha\, \frac{r^2}{n}}\, ,\qquad m \defeq \ceil*{\beta r}\, ,
\label{eq:ansatzes-L-m}
\ee
where $\alpha\in (0,1/2)$ and $\beta\in (0,1)$ will be determined later. Note that, since $n\ge  2$, we also have 
\bb
L\leq \alpha\,\frac{r^2}{n} + 1 \leq \alpha n + 1 < \frac{n}{2} + 1 \leq n\, ,
\ee
so that $L\in \{1,\ldots, n-1\}$, as we required. Plugging the values~\eqref{eq:ansatzes-L-m} into~\eqref{eq:first-bound-q_r}, we obtain
\bb
\max_{s\in \{1,\ldots,n\}} \abs*{q_r(s)}\ &\leq 2\ \exp\left[ - \beta r \arccosh\of*{\frac{n}{n-\alpha \frac{r^2}{n}}} \right] \\
&\eqt{(vi)}\ 2 \exp\left[ - \frac{\beta}{\sqrt\alpha}\, n \varphi\of*{\alpha \tfrac{r^2}{n^2}} \right] ,
\label{eq:second-bound-q_r}
\ee
where the auxiliary function $\varphi$ is defined by
\bb
\varphi(z) \defeq \sqrt{z} \arccosh\of*{\frac{1}{1-z}}\, .
\ee

Note that, for $0\le z<1$,
\begin{equation}
  \arccosh\of*{\frac1{1-z}}
  \ge
  \sqrt{2z}\, .
\end{equation}
Indeed, this is the standard bound $\arccosh(1+u)\ge\sqrt{2u/(1+u)}$ with $u=z/(1-z)$. Thus $\varphi(z)\ge \sqrt2\,z$ for $0\le z<1$, entailing that
\bb
\frac{\beta}{\sqrt\alpha}\,
  n\varphi\of*{\alpha\frac{r^2}{n^2}}
  &\ge
  \beta\sqrt{2\alpha}\,\frac{r^2}{n}
  =
  \kappa\, \frac{r^2}{n}\, , \qquad \kappa \defeq \beta \sqrt{2\alpha}\, .
\ee
Hence, the bound in~\eqref{eq:second-bound-q_r} becomes
\begin{equation}
  \max_{s\in \{1,\ldots,n\}} \abs*{q_r(s)}
  \le
  2e^{-\kappa r^2/n}\, .
\end{equation}

It remains to verify under what conditions on $\alpha$ and $\beta$ the above choices yield a polynomial $q_r$ of degree no larger than $r$. 
If $L=1$, then $q_r=B_m$ and $\deg q_r=m=\ceil*{\beta r}\le r$, because $\beta<1$. Assume now that $L\ge2$. Then,
\bb
\sum_{j=1}^{L-1}m_j\ &\leqt{(vii)}\ \sum_{j=1}^{L-1} \left( \frac{\pi}{4} \sqrt{\frac{n}{j}} +1 \right) \\ 
&\leq\ \frac{\pi}{4} \sqrt{n} \sum_{j=1}^{L-1}j^{-1/2} + L-1\\
&\leqt{(viii)}\ \frac\pi2\sqrt{n(L-1)} + L-1\\
&\lt{(ix)}\ \of*{\frac{\pi}{2}\, \sqrt\alpha + \alpha} r\, ,
\ee
where (vii)~comes from~\eqref{eq:polynomials_T_mj}, in~(viii) we employed the elementary estimate $\sum_{j=1}^{L-1}j^{-1/2} \leq \int_0^{L-1} x^{-1/2}\, \mathrm{d}x = 2\sqrt{L-1}$, and in~(ix) we observed that $L-1 < \alpha r^2/n \leq \alpha r$. Therefore,
\bb
\deg q_r &\le \ceil*{\beta r} + \sum_{j=1}^{L-1}m_j < \beta r+1 + \of*{\frac{\pi}{2}\, \sqrt\alpha + \alpha} r = \of*{\beta + \frac{\pi}{2}\, \sqrt\alpha + \alpha} r + 1\, .
\ee
Due to the strict inequality, to ensure that the left-hand side (which is an integer) is at most $r$ it suffices to impose that
\bb
\beta + \frac{\pi}{2}\, \sqrt\alpha + \alpha \leq 1\, .
\ee
We can now set
\bb
\alpha=\frac{1}{16}\, , \qquad \beta = 1 - \alpha - \frac{\pi}{2} \sqrt{\alpha} \approx 0.545\, , \qquad \kappa=\frac{15-2\pi}{32\sqrt{2}} \approx 0.1926\, ,
\ee
completing the proof.
\end{proof}

\begin{lemma}[(Hypergeometric expansion bound)] \label{lem:clean-hypergeometric-expansion}
Let $D\ge2$ and $1\le k\le n$. If $q$ is a polynomial of degree $r\le k/2$ on $\set{0,\ldots,n}$, then there are coefficients $b_0,\ldots,b_k$ such that
\begin{equation}
  \label{eq:q_b_relation}
  q\of s
  =
  \sum_{j=0}^k b_j \frac{\binom sj\binom{n-s}{k-j}}{\binom nk}
  \qquad
  \forall\ s \in \{0,1,\ldots,n\}.
\end{equation}
Furthermore, for $m_j=\binom{j+D-2}{D-2}$,
\begin{equation}\label{eq:clean-hypergeometric-expansion}
  \sum_{j=0}^k m_j\abs{b_j}^2
  \le
  \frac{k+1}{n+1}\,
  \binom{k+D-2}{D-2}\,
  \exp\of*{\frac{2r^2}{k}}\,
  \sum_{s=0}^n\abs{q\of s}^2\, .
\end{equation}
\end{lemma}

\begin{proof}
Define the hypergeometric transition linear operator $T_{n,k}:\C^{k+1}\to\C^{n+1}$ as
\begin{equation}
  T_{n,k}
  \defeq
  \sum_{s=0}^{n}\sum_{j=0}^{k} \frac{\binom sj\binom{n-s}{k-j}}{\binom nk} \ket{s}\bra{j}.
\end{equation}
We can interpret \cref{eq:q_b_relation} as the statement that the matrix $T_{n,k}$, 
acting on the vector $\ket{b}\in\C^{k+1}$ with coefficients $b_j$, produces the vector $\ket{q} \in \C^{n+1}$ with coefficients $q(s)$:
\begin{equation} \label{eq:relation_T_n_k_b_and_q}
  \ket{b} \defeq \sum_{j=0}^{k}b_j\ket{j}, \qquad
  \ket{q} \defeq \sum_{s=0}^{n}q\of s\ket{s}, \qquad 
  T_{n,k}\ket{b} = \ket{q}
\end{equation}
Thus the problem is to invert this map on the subspace of value vectors coming from degree-$r$ polynomials, with a quantitative bound on the coefficient norm.
The Hahn polynomial bases below give singular vectors for this map on that subspace.
Let $\ket{0},\ldots,\ket{k}$ be the standard basis of $\C^{k+1}$, and let $\ket{0},\ldots,\ket{n}$ be the standard basis of $\C^{n+1}$.

For $x\in\R$ and $i\in\N$, write $x_{\underline i}\defeq x(x-1)\cdots(x-i+1)$, with $x_{\underline 0}\defeq1$.
For an arbitrary non-negative integer $m$ and $0\le \ell\le m$, define the degree-$\ell$ Hahn polynomial with parameters $(0,0)$ on $\set{0,\ldots,m}$ by
\begin{equation}
  Q_\ell^{(m)}(x)
  =
  \sum_{i=0}^{\ell}
  (-1)^i
  \binom\ell i
  \binom{\ell+i}{i}
  \frac{x_{\underline i}}{m_{\underline i}},
  \qquad
  0\le x\le m\,.
\end{equation}
The standard Hahn orthogonality formula for parameters $(0,0)$, see \cite[Section 9.5]{KLS}, gives
\begin{equation}\label{eq:clean-hahn-norm}
  \sum_{x=0}^{m}Q_\ell^{(m)}(x)Q_{\ell'}^{(m)}(x)
  =
  N_\ell^{(m)}\delta_{\ell\ell'},
  \qquad
  N_\ell^{(m)}
  \defeq
  \sum_{x=0}^{m}\,\abs*{Q_\ell^{(m)}(x)}^2
  =
  \frac{m+1}{2\ell+1}
  \prod_{i=1}^{\ell}\frac{m+1+i}{m+1-i}.
\end{equation}
Set
\begin{equation}
  \ket{v_\ell^{(m)}}
  \defeq
  \frac1{\sqrt{N_\ell^{(m)}}}
  \sum_{x=0}^{m} Q_\ell^{(m)}(x)\, \ket{x}.
\end{equation}
Then $\braket{v_\ell^{(m)}}{v_{\ell'}^{(m)}}=\delta_{\ell\ell'}$.
We now compute the image of these vectors under $T_{n,k}$ using the classical expansion of Hahn polynomials in a discrete Bernstein basis~\cite{Sablonniere1993}.
We give a direct derivation in our notation.
For $0\le i\le k$,
\begin{align}
  \sum_{j=0}^k
  \frac{\binom sj\binom{n-s}{k-j}}{\binom nk}
  \frac{j_{\underline i}}{k_{\underline i}}
  &=
  \frac1{\binom nk}
  \sum_{j=i}^k
  \binom{s}{j}\binom{n-s}{k-j}
  \frac{j_{\underline i}}{k_{\underline i}}
  =
  \frac{s_{\underline i}}{k_{\underline i}\binom nk}
  \sum_{j=i}^k
  \binom{s-i}{j-i}\binom{n-s}{k-j}
  =
  \frac{s_{\underline i}}{k_{\underline i}\binom nk}
  \binom{n-i}{k-i}
  =
  \frac{s_{\underline i}}{n_{\underline i}}.
\end{align}
Substituting this identity term by term gives, for $0\le\ell\le k$,
\bb \label{eq:singular_vectors_T_n_k}
T_{n,k}\ket{v_\ell^{(k)}}
  &=
  \frac1{\sqrt{N_\ell^{(k)}}}
  \sum_{s=0}^{n}
  \of*{
    \sum_{j=0}^{k}
    \frac{\binom sj\binom{n-s}{k-j}}{\binom nk}
    Q_\ell^{(k)}(j)
  }
  \,\ket{s}
  =
  \frac1{\sqrt{N_\ell^{(k)}}}
  \sum_{s=0}^{n}Q_\ell^{(n)}(s)\,\ket{s}
  =
  \sqrt{\frac{N_\ell^{(n)}}{N_\ell^{(k)}}}\,
  \ket{v_\ell^{(n)}}.
\ee
Thus $\ket{v_\ell^{(k)}}$ and $\ket{v_\ell^{(n)}}$ are right and left singular vectors of $T_{n,k}$, with singular value $\sigma_\ell$ satisfying
\begin{equation}
  \sigma_\ell^2
  =
  \frac{N_\ell^{(n)}}{N_\ell^{(k)}}
  =
  \frac{n+1}{k+1}
  \prod_{i=1}^{\ell}
  \frac{k+1-i}{n+1-i}
  \frac{n+1+i}{k+1+i}.
\end{equation}
We can therefore write
\bb \label{eq:T_n_k_in_terms_of_singular_vectors}
T_{n,k}=
\sum_{\ell=0}^k
\sigma_\ell\,
\ket{v^{(n)}_\ell}\!\bra{v^{(k)}_\ell}, \qquad
T_{n,k}^{+}=
\sum_{\ell=0}^k
\sigma_\ell^{-1}
\ket{v^{(k)}_\ell}\!\bra{v^{(n)}_\ell},
\ee
where $T_{n,k}^{+}$ denotes the Moore--Penrose pseudoinverse.
It satisfies
\bb
T_{n,k}^{+}T_{n,k}^{\vphantom{+}} = \1_{\C^{k+1}}, \qquad T_{n,k}^{\vphantom{+}} T_{n,k}^{+} = \sum_{\ell=0}^k \ket{v_\ell^{(n)}}\!\bra{v_\ell^{(n)}}
\ee

We now bound the above singular values for the indices that will be used.
For $\ell\le k/2$,
\begin{align}
\ln \of*{\frac{n+1}{k+1}\, \sigma_\ell^{-2}}
&\le \sum_{i=1}^{\ell}\ln\frac{k+1+i}{k+1-i} = \sum_{i=1}^{\ell}\ln\left(1 + \frac{2i}{k+1-i}\right) \nonumber \\
&\le \sum_{i=1}^{\ell}\frac{2i}{k+1-i}  \\
&\le \frac{2}{k+1-\ell}\sum_{i=1}^{\ell} i = \frac{\ell(\ell+1)}{k+1-\ell}  \nonumber \\
&\le \frac{2\ell^2}{k}, \nonumber
\end{align}
where the last inequality can be verified by inspection. Hence,
\begin{equation}
  \sigma_\ell^{-2}
  \le
  \frac{k+1}{n+1}
  \exp\of*{\frac{2\ell^2}{k}} .
\end{equation}

Now, let us estimate $\sum_{j=0}^k m_j\abs{b_j}^2$.
Let $V_{n,r}\subseteq\C^{n+1}$ be the subspace of vectors $\ket v$ whose components $\braket{x}{v}$, for $x=0,\ldots,n$, are the values of a polynomial in $x$ of degree at most $r$.
Since $r\le n$, evaluation on the $n+1$ distinct points $0,\ldots,n$ is injective on polynomials of degree at most $r$, and hence $\dim V_{n,r}=r+1$.
Each $Q_\ell^{(n)}$ has degree $\ell$, so the orthonormal vectors $\ket{v_0^{(n)}},\ldots,\ket{v_r^{(n)}}$ form an orthonormal basis of $V_{n,r}$.
Since $\ket q\in V_{n,r}$, there are coefficients $a_0,\ldots,a_r$ such that
\bb
\ket{q}=\sum_{\ell=0}^r a_\ell\ket{v_\ell^{(n)}}, \qquad \sum_{\ell=0}^r\abs{a_\ell}^2=\sum_{s=0}^n\abs{q\of s}^2, \qquad 
\ket{b} = \sum_{j=0}^k b_j \ket{j} = T_{n,k}^{+} \ket{q} = \sum_{\ell=0}^r\sigma_\ell^{-1}a_\ell\ket{v_\ell^{(k)}}\, ,
\ee
where we employed~\eqref{eq:relation_T_n_k_b_and_q} and~\eqref{eq:T_n_k_in_terms_of_singular_vectors}.

Therefore,
\begin{equation}
  \sum_{j=0}^k\abs{b_j}^2
  =
  \sum_{\ell=0}^r\sigma_\ell^{-2}\abs{a_\ell}^2
  \le
  \frac{k+1}{n+1}
  \exp\of*{\frac{2r^2}{k}}
  \sum_{s=0}^n\abs{q\of s}^2.
\end{equation}
Since $m_j\le\binom{k+D-2}{D-2}$ for every $0\le j\le k$, this estimate implies \eqref{eq:clean-hypergeometric-expansion}.
\end{proof}

\begin{lemma}[(Product-vector Kraus approximant)]\label{lem:clean-diagonal-kraus-approximant}
Let $\HH\simeq\C^D$ with $D\ge2$, let $\ket\omega\in\HH$ be a unit vector, and write $\ket{\omega_n} = \Pi_n \ket\omega^{\otimes n}$, where $\Pi_n$ is the projector onto the symmetric subspace defined by~\eqref{eq:Pi_n}. Let $1\le k\le n$, let $\gamma>0$, and set $r=\ceil{\sqrt{2\gamma\ln2 / \kappa}\,n}$.
If $\gamma n>1$ and $r\le\frac{k}{2}$, then there exists a vector $c$ such that the compressed blurring channel on $\Sym^n\of{\HH}$, defined by~\eqref{eq:compressed-blurring}, has a Kraus approximant $K(c)$ with
\begin{align}
  K(c)\, \ket{\omega_n} &= \ket{\omega_n}\, ,
  \qquad
  \norm*{K(c)-\proj{\omega_n}}_\infty^2
  \le
  2^{-2\gamma n}\,, \\
  \label{eq:clean-diagonal-kraus-approximant}
  \norm{c}_2^2
  &\le
  D^k
  \of{k+1}
  \binom{k+D-2}{D-2}
  \exp\of*{2\frac{r^2}{k}}.
\end{align}
\end{lemma}

\begin{proof}
Since $r\le k/2$ and $k\le n$, we have $r\le n$, so \cref{lem:clean-discrete-delta-approximant} provides a polynomial $q_r$ of degree at most $r$ such that
\begin{equation}
  q_r\of0
  =
  1,
  \qquad
  \abs*{q_r\of s}
  \le
  2e^{-\kappa r^2/n}\quad \forall\ s\in\{1,\ldots, n\}\, .
\end{equation}
The definition of $r$ and the assumption $\gamma n>1$ imply
\begin{equation}
  2e^{-\kappa r^2/n}
  \le
  2e^{-2\gamma n\ln 2}
  =
  2^{1-2\gamma n}
  \le
  2^{-\gamma n}.
\end{equation}
Because $r\le k/2$, \cref{lem:clean-hypergeometric-expansion} provides coefficients $b_0,\ldots,b_k$ satisfying
\begin{equation}
  q_r\of s
  =
  \sum_{j=0}^{k}
  b_j
  \frac{\binom sj\binom{n-s}{k-j}}{\binom nk},
  \qquad
  0\le s\le n.
\end{equation}
With $m_j=\binom{j+D-2}{D-2}$, the same lemma also gives
\begin{equation}\label{eq:clean-delta-coefficient-intermediate}
  \sum_{j=0}^{k}m_j\abs{b_j}^2
  \le
  \frac{k+1}{n+1}
  \binom{k+D-2}{D-2}
  \exp\of*{2\frac{r^2}{k}}
  \sum_{s=0}^{n}\abs*{q_r\of s}^2.
\end{equation}
The preceding bound yields
\begin{equation}
  \sum_{s=0}^{n}\abs*{q_r\of s}^2
  \le
  1+n2^{-2\gamma n}
  \le
  n+1.
\end{equation}
Substituting this into \eqref{eq:clean-delta-coefficient-intermediate} gives
\begin{equation}\label{eq:clean-delta-coefficients}
  \sum_{j=0}^{k}m_j\abs{b_j}^2
  \le
  \of{k+1}
  \binom{k+D-2}{D-2}
  \exp\of*{2\frac{r^2}{k}}\, .
\end{equation}

We now lift this polynomial to a Kraus approximant.
Choose an orthonormal basis $\ket{1},\ldots,\ket{D}$ of $\HH$ with $\ket{1}=\ket{\omega}$.
We use the corresponding Dicke basis of $\Sym^n\of{\HH}$.
For every $\beta\in\N^D$ with $\abs\beta=k$, let $j(\beta)=\sum_{x=2}^D\beta_x$.
For $\beta,\beta'\in\N^D$ with $\abs\beta=\abs{\beta'}=k$, define the coefficient vector $c$ by
\begin{equation}
  c_{\beta',\beta}
  =
  \begin{cases}
    D^{k/2}b_{j(\beta)}, & \beta'=\beta,\\
    0, & \beta'\ne\beta.
  \end{cases}
\end{equation}
Define the corresponding Kraus combination operator by
\begin{equation}
  K(c)
  \defeq
  \sum_{\substack{\beta\in\N^D\\\abs\beta=k}}
  c_{\beta,\beta}K_{\beta,\beta}.
\end{equation}
Fix a Dicke vector $\ket\alpha$, and let $s=\sum_{x=2}^D\alpha_x$ be its total occupation outside the first basis vector.
The factor $D^{k/2}$ in $c_{\beta,\beta}$ cancels the factor $D^{-k/2}$ in the diagonal formula of \cref{lem:clean-symmetric-kraus}.
Grouping the resulting sum according to $j(\beta)$ gives
\begin{align}
  K(c)\ket\alpha
  &=
  \sum_{\substack{\beta\le\alpha\\\abs{\beta}=k}}
  b_{j(\beta)}
  \frac{\binom k\beta\binom{n-k}{\alpha-\beta}}{\binom n\alpha}
  \ket\alpha
  =
  \sum_{j=0}^{k}
  b_j
  \frac{\binom{s}{j}\binom{n-s}{k-j}}{\binom nk}
  \ket\alpha
  =
  q_r\of s\ket\alpha.
\end{align}
For the second equality, we used
\begin{equation}
  \frac{\binom k\beta\binom{n-k}{\alpha-\beta}}{\binom n\alpha}
  =
  \frac1{\binom nk}
  \prod_{x=1}^{D}\binom{\alpha_x}{\beta_x},
\end{equation}
followed by the multivariate Vandermonde identity at fixed $j(\beta)=j$.
Thus every Dicke vector is an eigenvector of $K(c)$ with eigenvalue $q_r(s)$.
For $\ket\alpha=\ket{\omega_n}$ one has $s=0$, and hence $K(c)\ket{\omega_n}=q_r(0)\ket{\omega_n}=\ket{\omega_n}$.
Every Dicke vector orthogonal to $\ket{\omega_n}$ has $s\ge1$ and therefore an eigenvalue of modulus at most $2^{-\gamma n}$.
Since $K(c)$ is diagonal in this basis, it follows that
\begin{equation}
  \norm*{K(c)-\proj{\omega_n}}_\infty^2
  \le
  2^{-2\gamma n}.
\end{equation}
It remains to bound the norm of the coefficient vector $c$.
For each $j$, the value $\beta_1=k-j$ is fixed, while the remaining $D-1$ entries of $\beta$ form a weak composition of $j$.
Consequently, there are exactly $m_j=\binom{j+D-2}{D-2}$ labels $\beta$ with $j(\beta)=j$.
Using \eqref{eq:clean-delta-coefficients}, we conclude that
\begin{equation}
  \norm{c}_2^2
  =
  D^k
  \sum_{j=0}^{k}m_j\abs{b_j}^2
  \le
  D^k
  \of{k+1}
  \binom{k+D-2}{D-2}
  \exp\of*{2\frac{r^2}{k}}\, ,
\end{equation}
completing the proof.
\end{proof}

\section{Alternative estimates based on the Mazzola--Sutter--Renner techniques} \label{sec:alternative_proof_MSR}

The purpose of this section is to use the methods of~\cite{MazzolaSutterRenner2026} to give an alternative estimate of the hypothesis testing relative entropy in the setting considered here.
We will see that this is quantitatively strictly worse than the one presented in this paper.

In this section only, all divergences and entropies are measured in nats.
Equivalently, the symbols $D$, $D_{\max}$, and $D_{\mathrm H}$ below are defined as before, but with natural logarithms and with order bounds written using $e^\lambda$.
The lower eigenvalue constant used in the MSR notation is $c_\tau/d$ in the notation of~\eqref{ax:free-full-rank}.

\begin{proposition} \label{prop:msr-bound}
Assume \eqref{ax:closed-convex}--\eqref{ax:permutation-invariant}.
Then, for all $s>0$ and all integers $n\ge1$, $0\le k<n$, and $r\ge0$ satisfying $16r\le n-k$, we have
\bb
\frac1n
D_{\mathrm H}^{e^{-s n}}\of{\rho^{\otimes n}\|\FF_n}
&\ge 
\frac1n
\rel{D}{\rho^{\otimes(n-k)}}{\FF_{n-k}}
-3h\of*{2\sqrt{\frac{r}{n-k}}}
-12\sqrt{\frac{r}{n-k}}\ln\frac{d^2}{c_\tau}\\
&\qquad
-\xi_n\ln\frac{d}{c_\tau}
-\frac{g\of{\xi_n}}{n}
-s,
\label{eq:msr-bound}
\ee
where
\begin{equation}
  \xi_n
  \defeq
  \sqrt2\,e^{-\frac{rk}{2n}+\frac{s n}{2}},
\end{equation}
and where
\begin{equation}
  g(x)
  \defeq
  (x+1)\ln(x+1)-x\ln x,
  \qquad
  h(x)
  \defeq
  -x\ln x-(1-x)\ln(1-x).
\end{equation}
\end{proposition}

Before we delve into the proof, it is instructive to examine what consequences can be drawn from~\eqref{eq:msr-bound}.
The bound can be seen to trivialise when $\xi_n\ge1$, simply because
\begin{equation}
\rel{D}{\rho^{\otimes(n-k)}}{\FF_{n-k}} \le (n-k)\ln\frac{d}{c_\tau} \le n\ln\frac{d}{c_\tau}.
\end{equation}
Hence, for a fixed $n$, we must choose $r$ and $k$ in such a way that the exponent appearing in the definition of $\xi_n$ is negative. Setting $\phi \coloneqq r/n$ and $\theta \coloneqq k/n$, this yields the inequality
\bb
\phi \theta > s\, ,
\ee
which implies immediately that
\bb
\frac{r}{n-k} = \frac{\phi}{1-\theta} > \phi\, .
\ee
We can repeat the same reasoning on the term $12\sqrt{r/(n-k)}\ln(d^2/c_\tau)$.
Replacing the lower eigenvalue constant of the MSR notation by $c_\tau/d$ only changes fixed prefactors, and the argument still forces
\begin{equation}
  \sqrt{\frac{r}{n-k}}
  \le
  \frac1{12}\, .
\end{equation}
In particular, $2 \sqrt{\frac{r}{n-k}}\le \frac1{6} < \frac1e$, which is the monotonicity threshold for the function $-x \ln(x)$. We can then write
\begin{equation}
3h\of*{2\sqrt{\frac{r}{n-k}}} \ge 6\sqrt{\phi}\ln\frac1{2\sqrt{\phi}} > 6\sqrt{\phi}\, .
\end{equation}

Using also 
\bb
\frac1n
\rel{D}{\rho^{\otimes(n-k)}}{\FF_{n-k}} \ge \frac{n-k}{n} D^\infty(\rho\|\FF) = (1-\theta) D^\infty(\rho\|\FF)\, ,
\ee
where the inequality is tight for asymptotically large $n$ and $k = k_n$, with $\lim_{n\to\infty} k_n/n = \theta$, we see that the right-hand side of~\eqref{eq:msr-bound} is asymptotically no larger than
\bb
(1-\theta) D^\infty(\rho\|\FF) - 6 \sqrt{\phi}\, ,
\ee
Maximising over $\theta$ and $\phi$ subject to the constraint $\phi\theta>s$ yields, in the best-case scenario,
\bb
D^\infty(\rho\|\FF) - 3^{5/3} \big(s\, D^\infty(\rho\|\FF)\big)^{1/3}\, ;
\ee
because of the spurious power $s^{1/3}$ on the right-hand side, this is strictly worse than our bound.

\begin{proof}[Proof of \cref{prop:msr-bound}]
We start by leveraging the usual~\cite[Theorem~12, Eq.~(96)]{RegulaLamiDatta} to write
\bb
\frac1n
D_{\mathrm H}^{e^{-s n}}\of{\rho^{\otimes n}\|\FF_n}
&\ge 
\frac1n
D_{\max}^{\sqrt{1-e^{-s n}}}\of{\rho^{\otimes n}\|\FF_n}
+\frac1n\ln\frac1{1-e^{-s n}}\\
&\ge 
\frac1n
D_{\max}^{\sqrt{1-e^{-s n}}}\of{\rho^{\otimes n}\|\FF_n}
\eqdef
\lambda_n .
\label{eq:msr-bound-proof-1}
\ee
By definition of purified distance, we can find a smoothed state $\bar\rho_n$ and a free state $\sigma_n\in\FF_n$ such that
\bb
F^2\of{\bar\rho_n,\rho^{\otimes n}}
\ge
e^{-s n},
\qquad
\bar\rho_n
\le
e^{n\lambda_n}\sigma_n,
\qquad
\sigma_n\in\FF_n.
\label{eq:msr-bound-proof-2}
\ee
Without loss of generality, we can assume that both $\bar\rho_n$ and $\sigma_n$ are permutation invariant.
Let $\ket\theta\in\HH\otimes\mathcal R$ be a purification of $\rho$, with $\mathcal R\simeq\HH$.
By~\cite[Lemma~III.4]{BrandaoPlenio}, equivalently by \cref{lem:clean-purification-reduction}, we can now find a permutation-symmetric purification $\ket{\Psi_n}\in\Sym^n\of{\HH\otimes\mathcal R}$ of $\bar\rho_n$ such that
\bb
\abs*{\braket{\Psi_n}{\theta^{\otimes n}}}
=
F\of{\bar\rho_n,\rho^{\otimes n}}
\ge
e^{-s n/2}.
\label{eq:msr-bound-proof-3}
\ee

Fix integers $k,r$ satisfying $16r\le n-k$.
Since this implies $r\le n-k$, the exponential de Finetti theorem~\cite{RennerPhD, Renner2007} in the form~\cite[Lemma~III.5]{BrandaoPlenio}, see also~\cite[Lemma~A.7]{MazzolaSutterRenner2026}, applied with $m=k$, shows the existence of a state $\Upsilon_{n-k}$ with two properties.
First,
\bb
\Upsilon_{n-k}
\le
\frac{\Tr_k\proj{\Psi_n}}
     {\abs*{\braket{\Psi_n}{\theta^{\otimes n}}}^2},
\label{eq:msr-bound-proof-4}
\ee
where $\Tr_k$ denotes the partial trace over $k$ subsystems.
Second, there exists an almost i.i.d.\ state
\begin{equation}
  \Phi_{n-k,r}
  \in
  \mathrm{S}\of*{(\HH\otimes\mathcal R)^{\otimes(n-k)},\,\ket\theta^{\otimes(n-k-r)}}
\end{equation}
such that
\bb
\frac12
\norm*{\Upsilon_{n-k}-\Phi_{n-k,r}}_1
\le
\frac{\sqrt2\,e^{-\frac{rk}{2n}}}
     {\abs*{\braket{\Psi_n}{\theta^{\otimes n}}}}.
\label{eq:msr-bound-proof-5}
\ee

Set
\bb
\upsilon_{n-k}
\defeq
\Tr_{\mathcal R^{\otimes(n-k)}}\Upsilon_{n-k},
\qquad
\phi_{n-k,r}
\defeq
\Tr_{\mathcal R^{\otimes(n-k)}}\Phi_{n-k,r}.
\ee
Then $\phi_{n-k,r}$ is almost i.i.d.\ along $\rho$ by construction.
Since $16r\le n-k$, the stability of the relative entropy of resource for almost i.i.d.\ states~\cite[Corollary~3.3]{MazzolaSutterRenner2026}, applied to $n-k$ systems, gives
\bb
\abs*{
\frac1nD\of{\phi_{n-k,r}\|\FF_{n-k}}
-
\frac1nD\of{\rho^{\otimes(n-k)}\|\FF_{n-k}}
}
\le
3h\of*{2\sqrt{\frac{r}{n-k}}}
+12\sqrt{\frac{r}{n-k}}\ln\frac{d^2}{c_\tau}.
\label{eq:msr-bound-proof-7}
\ee

Now, taking the partial traces of~\eqref{eq:msr-bound-proof-4} and~\eqref{eq:msr-bound-proof-5} over the whole purifying system $\mathcal R^{\otimes(n-k)}$, and combining with the estimates in~\eqref{eq:msr-bound-proof-2} and~\eqref{eq:msr-bound-proof-3}, we see that
\begin{align}
\upsilon_{n-k}
&\le
e^{s n}\Tr_k\bar\rho_n
\le
e^{n(\lambda_n+s)}\bar\sigma_{n-k},
\label{eq:msr-bound-proof-8}\\
\frac12\norm*{\upsilon_{n-k}-\phi_{n-k,r}}_1
&\le
\sqrt2\,e^{-\frac{rk}{2n}+\frac{s n}{2}}
=
\xi_n,
\label{eq:msr-bound-proof-9}
\end{align}
where $\bar\sigma_{n-k}\defeq\Tr_k\sigma_n$.
Note that $\bar\sigma_{n-k}\in\FF_{n-k}$ by the Brand\~ao--Plenio axioms.

In particular,
\bb
n\of{\lambda_n+s}
&\geqt{(i)}
D_{\max}\of{\upsilon_{n-k}\|\FF_{n-k}}\\
&\geqt{(ii)}
D\of{\upsilon_{n-k}\|\FF_{n-k}}\\
&\geqt{(iii)}
D\of{\phi_{n-k,r}\|\FF_{n-k}}
-n\xi_n\ln\frac{d}{c_\tau}
-g\of{\xi_n}\\
&\geqt{(iv)}
D\of{\rho^{\otimes(n-k)}\|\FF_{n-k}}
-3n\,h\of*{2\sqrt{\frac{r}{n-k}}}
-12n\sqrt{\frac{r}{n-k}}\ln\frac{d^2}{c_\tau}\\
&\qquad
-n\xi_n\ln\frac{d}{c_\tau}
-g\of{\xi_n}.
\label{eq:msr-bound-proof-10}
\ee
Here, (i) comes from~\eqref{eq:msr-bound-proof-8}.
In (ii), we simply observed that the max-relative entropy upper bounds the Umegaki relative entropy.
In (iii), we used the standard continuity of the relative entropy of resource~\cite{Donald1999, MatthiasPhD, WinterContinuity, LamiGQSL} together with~\eqref{eq:msr-bound-proof-9}.
Finally, (iv) comes from~\eqref{eq:msr-bound-proof-7}.
Plugging~\eqref{eq:msr-bound-proof-10} into~\eqref{eq:msr-bound-proof-1} completes the proof.
\end{proof}

\end{document}

%% file: macros.tex
\newtheorem*{theorem*}{Theorem}
\newtheorem{theorem}{Theorem}
\newtheorem{proposition}[theorem]{Proposition}
\newtheorem{lemma}[theorem]{Lemma}
\newtheorem{definition}{Definition}

\newtheorem{remark}{Remark}

\counterwithin*{equation}{part}
\counterwithin*{theorem}{part}
\counterwithin*{definition}{part}
\counterwithin*{remark}{part}
\counterwithin*{figure}{part}

\makeatletter
\@ifpackageloaded{hyperref}{%
}{}
\makeatother

\makeatletter
\def\thmhead@plain#1#2#3{%
  \thmname{#1}\thmnumber{\@ifnotempty{#1}{ }\@upn{#2}}%
  \thmnote{ {\the\thm@notefont#3}}}
\let\thmhead\thmhead@plain
\makeatother

\crefname{lemma}{Lemma}{Lemmas}
\crefname{definition}{Definition}{Definitions}
\crefname{theorem}{Theorem}{Theorems}
\crefname{conjecture}{Conjecture}{Conjectures}
\crefname{section}{Section}{Sections}
\crefname{claim}{Claim}{Claims}
\crefname{appendix}{Appendix}{Appendices}
\crefname{figure}{Fig.}{Figs.}
\crefname{table}{Table}{Tables}
\crefname{proposition}{Proposition}{Propositions}
\crefname{corollary}{Corollary}{Corollaries}
\crefname{example}{Example}{Examples}
\crefname{remark}{Remark}{Remarks}

\providecommand\given{}
\newcommand\SetSymbol[1][]{%
    \nonscript\:#1\vert
    \allowbreak
    \nonscript\:
    \mathopen{}}
\DeclarePairedDelimiterX\Set[1]\{\}{%
    \renewcommand\given{\SetSymbol[\delimsize]}
    #1
}

\DeclarePairedDelimiter{\set}{\lbrace}{\rbrace}
\DeclarePairedDelimiter{\abs}{\lvert}{\rvert}
\DeclarePairedDelimiter{\norm}{\lVert}{\rVert}
\DeclarePairedDelimiter{\floor}{\lfloor}{\rfloor}
\DeclarePairedDelimiter{\ceil}{\lceil}{\rceil}
\DeclarePairedDelimiter{\of}{\lparen}{\rparen}
\DeclarePairedDelimiter{\sof}{\lbrack}{\rbrack}

\newcommand{\defeq}{\vcentcolon=}
\newcommand{\eqdef}{=\vcentcolon}
\renewcommand{\leq}{\leqslant}

\providecommand{\bra}[1]{}
\providecommand{\ket}[1]{}
\providecommand{\braket}[2]{}
\providecommand{\proj}[1]{}
\renewcommand{\bra}[1]{\langle{#1}\rvert}
\renewcommand{\ket}[1]{\lvert{#1}\rangle}
\renewcommand{\braket}[2]{\langle{#1}|{#2}\rangle}
\newcommand{\ketbra}[2]{\ket{#1}\bra{#2}}
\renewcommand{\proj}[1]{\ketbra{#1}{#1}}

\newcommand{\x}{\otimes}

\newcommand{\C}{\mathbb{C}} %
\newcommand{\N}{\mathbb{N}} %
\newcommand{\R}{\mathbb{R}} %

\newcommand{\1}{\mathbbm{1}} %

\DeclareMathOperator{\Sym}{Sym} %
\DeclareMathOperator{\D}{D} %
\DeclareMathOperator{\Tr}{Tr} %
\DeclareMathOperator{\supp}{supp} %
\newcommand{\e}{\varepsilon}

\newcommand\restr[2]{{%
  \left.\kern-\nulldelimiterspace %
  #1 %
  \vphantom{\big|} %
  \right|_{#2} %
  }}

\DeclareMathOperator{\arccosh}{arccosh}
\DeclareMathOperator{\pr}{Pr}

\newcommand{\deff}[1]{\textbf{\textit{#1}}}
\newcommand{\rel}[3]{#1\big(#2\,\big\|\,#3\big)}

\newcommand{\stein}{\mathrm{Stein}}
\newcommand{\hoeff}{\mathrm{Hoeffding}}
\newcommand{\revh}{\mathrm{RevHoeff}}

\newcommand{\sep}{\mathrm{SEP}}

\newcommand{\HH}{\mathcal{H}}

\newcommand{\FF}{\mathcal{F}}
\newcommand{\BB}{\mathcal{B}}
\newcommand{\TT}{\mathcal{T}}
\newcommand{\XX}{\mathcal{X}}

\newcommand{\dhh}[1]{D_{\mathrm{H}}^{#1}}

\newcommand{\bb}{\begin{equation}\begin{aligned}\hspace{0pt}}
\newcommand{\ee}{\end{aligned}\end{equation}}

\newcommand{\eqt}[1]{\stackrel{\mathclap{\scriptsize \mbox{#1}}}{=}}
\newcommand{\leqt}[1]{\stackrel{\mathclap{\scriptsize \mbox{#1}}}{\le}}
\newcommand{\geqt}[1]{\stackrel{\mathclap{\scriptsize \mbox{#1}}}{\ge}}
\newcommand{\lt}[1]{\stackrel{\mathclap{\scriptsize \mbox{#1}}}{<}}

\stackMath

\newcommand{\fakepart}[1]{
 \par\refstepcounter{part}
  \sectionmark{#1}
}